\documentclass[aps,pra,twocolumn]{revtex4-1} 

\usepackage{verbatim}
\usepackage{dsfont} 
\usepackage{amsmath}
\usepackage{amsthm} 
\usepackage{enumerate}
\usepackage{graphicx}
\usepackage{mathtools} 
\usepackage{bm} 
\usepackage{amssymb} 
\usepackage{hyperref} 
\hypersetup{
	colorlinks=true,
	linkcolor=blue,
	citecolor=blue,
	unicode=true,          
}
\usepackage[capitalize,nameinlink]{cleveref}

\usepackage{blindtext}

\newtheorem{result}{Result}
\newtheorem{lemma}{Lemma}

\newcommand{\braket}[1]{\left\langle #1 \right\rangle}

\DeclareMathOperator{\tr}{Tr}
\DeclareMathOperator{\Tr}{Tr}
\newcommand{\norm}[1]{\left\lVert#1\right\rVert}

\newcommand{\LL}{\mathcal{L}}
\newcommand{\hi}{\mathcal{H}} 

\newcommand{\M}{\mathsf{M}}

\usepackage[bbgreekl]{mathbbol}
\DeclareSymbolFontAlphabet{\mathbb}{AMSb}
\DeclareSymbolFontAlphabet{\mathbbl}{bbold}
\DeclareMathSymbol{\bbepsilon}{\mathord}{bbold}{"0F}
\newcommand{\1}{\mathds{1}}
\usepackage{mathrsfs}
\newcommand{\eff} {\mathscr{E}}
\newcommand{\ang}[1]{\left\langle #1 \right\rangle}

\usepackage{calc} 
\usepackage{accents}
\newcommand{\dbtilde}[1]{\accentset{\approx}{#1}}

\makeatletter
\newcommand{\oset}[3][0ex]{%
	\mathrel{\mathop{#3}\limits^{
			\vbox to#1{\kern-2\ex@
				\hbox{$\scriptstyle#2$}\vss}}}}
\makeatother

\usepackage{diagbox}
\usepackage{textpos}
\begin{document}

\title{Multi-object operational tasks for measurement incompatibility}

\author{Andr\'es F. Ducuara$^{1,2}$} 
\email[]{andres.ducuara@yukawa.kyoto-u.ac.jp}

\author{Ryo Takakura$^{3}$}
\email[]{takakura.ryo.qiqb@osaka-u.ac.jp}

\author{Fernando J. Hernandez$^{4}$}
\email[]{fhernandezortega@correo.unicordoba.edu.co}

\author{Cristian E. Susa$^{4}$}
\email[]{cristiansusa@correo.unicordoba.edu.co}

\affiliation{$^{1}$Yukawa Institute for Theoretical Physics, Kyoto University, Kitashirakawa Oiwakecho, Sakyo-ku, Kyoto 606-8502, Japan
    \looseness=-1}

\affiliation{$^{2}$Center for Gravitational Physics and Quantum Information, Yukawa Institute for Theoretical Physics, Kyoto University
    \looseness=-1} 

\affiliation{$^{3}$Center for Quantum Information and Quantum Biology, Osaka University, Japan
	\looseness=-1}

\affiliation{$^{4}$Department of Physics and Electronics, University of C\'ordoba, 230002 Monter\'ia, Colombia\looseness=-1}

\date{\today}
\begin{abstract}
    We introduce \textit{multi-object} operational tasks for measurement incompatibility in the form of multi-object quantum subchannel discrimination and exclusion games with prior information, where a player can simultaneously harness the resources contained within both a quantum state and a set of measurements. We show that any fully or partially resourceful pair of objects is useful for a suitably chosen multi-object subchannel discrimination and exclusion game with prior information. The advantage provided by a fully or partially resourceful object against all possible fully free objects in such a game can be quantified in a \textit{multiplicative} manner by the resource quantifiers of generalised robustness and weight of resource for discrimination and exclusion games, respectively. These results hold for arbitrary properties of quantum states as well as for arbitrary properties of sets of measurements closed under classical pre and post-processing and, consequently, include measurement incompatibility as a particular case. We furthermore show that these results are not exclusive to quantum theory, but that can also be extended to the realm of general probabilistic theories.
\end{abstract} 
\maketitle

\begin{textblock*}{3cm}(17cm,-10.5cm)
  \footnotesize YITP-24-176
\end{textblock*}
\vspace{-1cm}
\section{Introduction}
	
The theoretical framework of quantum resource theories (QRTs) \cite{RT_review} has consolidated itself during the past two decades as a fruitful approach to quantum information where quantum properties of physical systems are harnessed for the benefit of operational tasks. A quantum resource theory can be specified by first defining a quantum \textit{object} of interest, followed by one property of such objects to be exploited as a \textit{resource} \cite{RT_review}. The two most explored QRTs are arguably those of quantum states \cite{QRTE1, QRTE2} and measurements \cite{RT_measurements0, RT_measurements1, RT_measurements2}. QRTs of states explore desirable properties such as: entanglement \cite{QRTE2}, coherence \cite{RT_coherence}, asymmetry \cite{RT_coherence}, superposition \cite{RT_superposition}, purity \cite{RoP}, magic \cite{RT_magic}, amongst many others \cite{RT_nongaussianity, RT_nonmarkovianity1, RT_nonmarkovianity2, namit, RT_thermodynamics, RT_RF}. Similarly, QRTs of measurements explore properties such as: entanglement \cite{RT_measurements1}, coherence \cite{RT_measurements1}, informativeness \cite{RoM}, and non-projective simulability \cite{RT_PS}. There exist however additional desirable resources contained within more general types of objects such as: sets of measurements \cite{invitation_2016, review_incompatibility}, behaviours or boxes \cite{RT_nonlocality, RT_noncontextuality}, steering assemblages \cite{RT_steering}, teleportation assemblages \cite{RoT}, and channels \cite{RT_channels1, RT_channels2, Plenio, Pirandola}, amongst many others \cite{RT_review, Schmid1, Schmid2, Schmid3, Schmid4}.  In this work we focus on QRTs of \textit{sets of measurements}, with the resource of \textit{measurement incompatibility} in particular. Measurement incompatibility \cite{invitation_2016, review_incompatibility} is a property lying at the foundations of quantum mechanics which acts as a parent resource for the properties of Bell-nonlocality \cite{review_nonlocality_2014} and EPR-steering \cite{review_steering_2017, review_steering_2020} and, as a consequence of this, acts as a prerequisite for practical applications relying on fully device-independent as well as semi device-independent protocols \cite{VS_nonlocality, DI0, DI1, DI2, DI3}.

A broad research area of interest within QRTs concerns with the development of operational tasks harnessing the resources contained within quantum objects. In the particular case of the QRT of incompatible sets of measurements, it has been found that incompatibility serves as fuel for \textit{single-object} quantum state discrimination tasks \cite{paul_ivan_dani_2019, Teiko_2019, Uola_2019, RoI_local, QP1, QP2}. It has however also been pointed out, that there exist operational tasks that can simultaneously exploit the resources contained within \textit{multiple} objects, or \textit{multi-object} operational tasks for short \cite{MO1, MO2, MR1, MR2}. 
These \textit{multi-object} tasks are currently specifically tailored to exploit the resources contained within states and individual measurements and therefore, they currently do not accommodate for desirable properties of \textit{sets} of measurements such as measurement incompatibility. In this work we address this shortcoming by introducing \textit{multi-object} operational tasks for measurement incompatibility in the form of \textit{multi-object} quantum subchannel discrimination and exclusion games with prior information. We characterise the advantage provided by incompatible sets of measurements when playing these tasks in terms of resource quantifiers, and we furthermore extend these results to the framework of general probabilistic theories \cite{hardy2001quantum, PhysRevLett.99.240501, barnum2012teleportation, PhysRevA.75.032304, PhysRevA.81.062348, PhysRevA.84.012311, Masanes_Muller_2011}.

This document is organised as follows. We first address measurement incompatibility and classical pre and post-processing. We then introduce multi-object operational tasks for states and sets of measurements. We then introduce multi-object operational tasks in the form of discrimination games, exclusion games, and furthermore extend these results to the realm of general probabilistic theories. We end up with conclusions and perspectives.
	
\section{Measurement incompatibility and classical pre and post-processing}
\label{s:s2}

Let $\LL(\hi)$, $\LL_S(\hi)$, and $\LL_S^+(\hi)$, be the sets of all bounded, self-adjoint, and positive semi-definite operators, respectively, on a finite-dimensional complex Hilbert space $\hi$. A quantum {state}, or density operator, is a trace one positive semi-definite operator. The set of all states is denoted as $\mathcal{D}(\hi) \coloneqq \{\rho \in \LL_S^+(\hi) \mid \tr[\rho] = 1\}$. A quantum {measurement} or normalised positive operator-valued measure (POVM) is a set of positive semi-definite operators $\{M_a\}_{a=1}^l$ satisfying $\sum_{a\in A} M_{a} = \mathds{1}$, with $A = \{1,\ldots,l\}$, and $\1$ the identity operator on $\hi$. More formally, a measurement $\mathsf{M}$ can be stated to be a map from a $\sigma$-algebra on an outcome set $A = \{1,\ldots,l\}$ to $\LL_S^+(\hi)$, and so $\M = \{M_a\}_{a=1}^l$. The operator $M_a=\M(\{a\}) \in \LL_S^+(\hi)$ is denoted as the {POVM element} or {effect} corresponding to the specific outcome $a \in A$. A \textit{POVM set} is a set of POVMs $\{\mathsf{M}_x\}_{x=1}^\kappa$, with each POVM having the same outcome space $A$ and so, for convenience, we write $\{\mathsf{M}_x\}_{x=1}^\kappa \equiv \mathbb{M}_{A|X}=\{M_{a|x}\}_{a,x}$, where $M_{a|x}=\M_x(\{a\})$, and thus $\sum_{a\in A} M_{a|x} = \mathds{1}$, $\forall x\in\{1,\ldots, \kappa\}$. We will use the further simplified notation $\{M_{a|x}\} \equiv \{M_{a|x}\}_{a,x}$.

Let us also invoke the following notation from probability theory. Let ($X, G,...$) be random variables  on a finite alphabet $\mathcal{X}$, and the probability mass function (PMF) of a random variable $X$ represented as $p_X$ satisfying: $p_X(x)\geq 0$, $\forall x \in \mathcal{X}$, and $\sum_{x \in \mathcal{X}} p_X(x) = 1$. For simplicity, we address $p_X(x)$ as $p(x)$ when evaluating, and omit the alphabet when summing. Joint and conditional PMFs are denoted as $p_{XG}$, $p_{G|X}$, respectively. With this notation in place, let us now address the property of \textit{incompatibility} of POVM sets. A POVM set $\mathbb{M}_{A|X} = \{M_{a|x}\}$ is {compatible} whenever there exist a parent POVM $\mathbb{G}=\{G_\lambda\}$ and a PMF $p_{A|X\Lambda}$ such that:
\begin{align}
    M_{a|x}
    =
    \sum_{\lambda}
    p(a|x,\lambda)\,
    G_\lambda,
    \hspace{0.5cm}
    \forall a,x,
    \label{eq:comp}
\end{align}
and it is called incompatible otherwise. This property of POVM sets is commonly also refereed to as measurement compatibility or joint measurability. For simplicity, in this work we stick to measurement (in)compatibility or simply (in)compatibility. It will be useful to introduce the operation of \textit{simulability} of POVM sets. We say that a POVM set $\mathbb{N}_{B|Y} = \{N_{b|y}\}$, $b\in \{1,...,m\}$, $y\in \{1,...,\tau\}$ is {simulable by} the POVM set $\mathbb{M}_{A|X} = \{M_{a|x}\}$, $a\in \{1,...,l\}$, $x\in \{1,...,\kappa\}$ whenever there exist a triplet of PMFs $q_Z$, $ r_{X|YZ}$, and $s_{B|AYZ}$  such that  \cite{simulability}:
\begin{align}
    N_{b|y}
    =
    \sum_{a,x,z}
    s(b|a,y,z)\,
    M_{a|x}\,
    r(x|y,z)\,
    q(z)
    ,
    \hspace{0.5cm}
    \forall b,y.
    \label{eq:CP}
\end{align}
Simulability of POVM sets can be thought of as composed of a classical pre-processing (CPreP) stage, represented by the PMFs $q_Z$ and $r_{X|YZ}$, and a classical post-processing (CPosP) stage, represented by the PMF $s_{B|AYZ}$. The triplet of PMFs allowing the simulation are going to be refereed to as the set of \textit{strategies} $\mathcal{S} = \{q_Z, r_{X|YZ}, s_{B|AYZ}\}$. In \autoref{fig:fig1} we illustrate the simulability of POVM sets. One can check that the simulability of POVM sets defines a partial order for POVM sets and therefore, this motivates the notation $\mathbb{N}_{B|Y} \preceq \mathbb{M}_{A|X}$, meaning that $\mathbb{N}_{B|Y}$ is simulable by $\mathbb{M}_{A|X}$. It is straightforward to check that CPreP and CPosP are operations that take {compatible} POVM sets into {compatible} POVM sets and therefore, in this sense, we say measurement compatibility is closed under simulation of POVM sets. 

\begin{figure}
    \centering
    \includegraphics[scale=1.1]{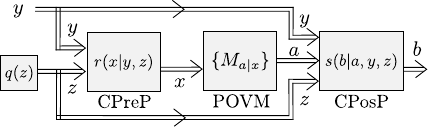}
    \vspace{-0.2cm}
    \caption{
        Classical pre-processing (CPreP) and post-processing (CPosP) of the POVM set $\mathbb{M}_{A|X} = \{M_{a|x}\}$ into a new POVM set $\mathbb{N}_{B|Y} = \{N_{b|y}\}$ with each $N_{b|y}$ as in \eqref{eq:CP}. A random variable $z$ is generated according to the PMF $q_Z$. The pre-processing stage consists on choosing a classical input $x$ according to a PMF which depends on the classical input $y$ and $z$ as $r_{X|YZ}$. The post-processing consists on processing the measurement outcome $a$ according to a PMF which depends on $y$ and $z$ as $s_{B|AYZ}$. We obtain the pre and post-processed set of measurements $\mathbb{N}_{B|Y} = \{N_{b|y}\}$.
    }
    \label{fig:fig1}
\end{figure}
We now move on to introducing the concept of an ensemble of instruments. A quantum {channel} is a map $\phi\colon \LL(\hi)\to \LL(\hi)$ that is completely positive (CP) and trace-preserving (TP). A quantum {subchannel} is a map $\phi\colon \LL(\hi)\to \LL(\hi)$ that is completely positive (CP) and trace-nonincreasing (TNI).  A quantum instrument is a set of subchannels $\Phi = \{\phi_b\}_{b=1}^m$ such that $\sum_{b=1}^m\phi_{b}$ is a channel, $\phi_b = \Phi(\{b\})$, $\forall b\in B$. Instruments are mathematical representations of quantum measurement processes. Given an instrument $\Phi = \{ \phi_b \}_{b=1}^m$, we can construct a measurement model such that the output probabilities and post-measurement states for a given target state $\rho$ are respectively $\{\tr[\phi_b(\rho)]\}_b$ and $\{\phi_b(\rho)/\tr[\phi_b(\rho)]\}_b$ \cite{Busch_quantummeasurement, heinosaari_ziman_2011}. An {instrument set} is a set of instruments $\{\Phi_y\}_{y=1}^\tau$ with each instrument having the same outcome set $B$. Similarly to the case of POVM sets, we can alternatively write an instrument set as $\Phi_{B|Y} = \{\phi_{b|y}\}_{b,y}$, where $\phi_{b|y} = \Phi_y(\{b\})$, and thus $\sum_{b=1}^m\phi_{b}$ is a channel $\forall y\in\{1,\ldots,\tau\}$, $B=\{1,\ldots\,m\}$. 
Finally, an {ensemble of instruments} is a pair $(p_Y,\Phi_{B|Y})$ with $\Phi_{B|Y}$ an instrument set and $p_Y$ a PMF. We now introduce multi-object operational tasks which are meant to be played with two quantum objects, a quantum state and a POVM set.	
	
\section{Multi-object operational tasks for states and POVM sets}
\label{s:s3}

We now introduce the operational task of multi-object quantum subchannel discrimination game with prior information (QScD-PI), as a generalisation of the single-object state discrimination tasks with prior information first introduced in \cite{Teiko_2018}. We illustrate this multi-object game in \autoref{fig:fig2} and describe it as follows. The game is played by two parties, Alice (the player) and Bob (the referee). Alice is in possession of two quantum objects: a state $\rho$ and a set of $\kappa$ POVMs $\{\M_x\}_{x=1}^{\kappa} \equiv \mathbb{M}_{A|X} =\{M_{a|x}\}$. Here all the POVMs are considered with the same outcome set $A=\{1,\ldots,l\}$ and $M_{a|x}\in\mathcal {L}_S^+(\hi)$, $\forall a,x$. Bob on the other hand is in possession of a set of $\tau$ instruments $\{\Phi_y\}_{y=1}^\tau \equiv \Phi_{B|Y} = \{\phi_{b|y}\}$. All instruments are considered with the same outcome set $B=\{1,\ldots,m\}$, and $\phi_{b|y} \colon \mathcal {L}(\hi)\to\mathcal {L}(\hi)$, $\forall b,y$. The first step of the game is for Alice (player) to prepare a quantum state $\rho$ and send it to Bob (referee). Bob  then proceeds to implement one of the instruments $\Phi_{B|Y}$, say $\Phi_y=\{\phi_{b|y}\}_y$, according to the PMF $q_Y$, on the set $Y=\{1,\ldots,\tau\}$, for which we assume that $q(y) \neq 0$, $\forall y\in Y$. Bob then conducts the measurement associated to $\Phi_y$ on $\rho$, observing an outcome $b \in B$ and a post-measurement state $\rho_{b|y} \coloneqq \phi_{b|y}(\rho) / \tr[\phi_{b|y}(\rho)]$. After Bob's measurement, Alice is informed of Bob's choice $y$ (prior information), and is also sent the state $\rho_{b|y}$. Alice's goal is to correctly guess the label $b$ of the sub-channel $\phi_{b|y}$. In order to do this, Alice first generates a random variable $z$ according to a PMF $q_Z$, and uses this to determine the choice of measurement $\M_x$ with a probability $r(x|y,z)$. Alice then proceeds to measure $\M_x=\{M_{a|x}\}$ on $\rho_{b|y}$, and observes an outcome $a$. Alice is allowed to classically post-process this measurement outcome to generate a guess $g\in B$ of $b$, according to $s(g|a,y,z)$. Finally, Alice sends her guess $g$ to Bob, and wins the game whenever $g=b$. We address this task as \textit{multi-object} quantum subchannel discrimination with prior information (QScD-PI). 
The maximum probability of {success} in such a task is given by:
\begin{align}
    P^{\rm D}_{\rm succ}
    (
    &
    p_Y,
    \Phi_{B|Y},
    \rho,
    \mathbb{M}_{A|X}
    )
    \coloneqq
    \max_{\mathcal{S}}
    \hspace{-0.2cm}
    \sum_{g,a,x,\mu,b,y}
    \hspace{-0.3cm}
    \delta_{g,b}\,
    s(g|a,y,z)
    \nonumber
    \\
    &
    \times
    \tr[
    M_{a|x}
    \phi_{b|y}
    (\rho)
    ]\,
    r(x|y,z)\,
    q(z)\,
    p(y),\label{eq:Psuc_Q}
\end{align}
with the maximisation over all possible strategies $\mathcal{S} = \{q_Z, r_{X|YZ}, s_{G|AYZ}\}$. In a multi-object quantum subchannel \emph{exclusion} game with prior information (QScE-PI) on the other hand, the goal is for Alice (player) to output a guess $g \in \{1,...,m\}$ for a subchannel that did {not} take place, that is, Alice succeeds at the game if $g \neq b$ and fails when $g=b$. The minimum probability of {error} in quantum subchannel exclusion with prior information is:
\begin{align}
    P^{\rm E}_{\rm err}
    (
    &
    p_Y,
    \Phi_{B|Y},
    \rho,
    \mathbb{M}_{A|X}
    )
    \coloneqq
    \min_{\mathcal{S}}
    \hspace{-0.2cm}
    \sum_{g,a,x,\mu,b,y}
    \hspace{-0.3cm}
    \delta_{g,b}\,
    s(g|a,y,z)
    \nonumber
    \\
    &
    \times
    \tr[
    M_{a|x}
    \phi_{b|y}
    (\rho)
    ]\,
    r(x|y,z)\,
    q(z)\,
    p(y),\label{eq:Perr_Q}
\end{align}
with the minimisation over all possible strategies $\mathcal{S} = \{q_Z, r_{X|YZ}, s_{G|AYZ}\}$. A multi-object quantum subchannel discrimination/exclusion game is specified by the \textit{ensemble of instruments} $\{p_Y,\Phi_{B|Y}\}$. A key point to remark here is that, the object of interest is now the state-POVM set pair $(\rho, \mathbb{M}_{A|X})$, as opposed to the POVM set alone, as it is the case in standard single-object state discrimination tasks \cite{Teiko_2018}. We now analyse these multi-object tasks from the point of view of resource theories.

\begin{figure}
    \centering
    \includegraphics[scale=0.97]{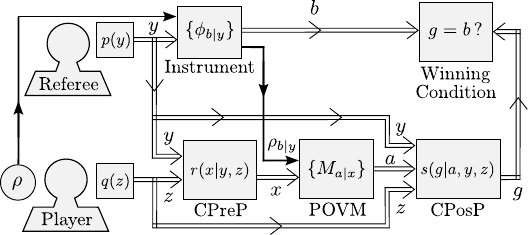}
    \caption{
        Multi-object quantum subchannel discrimination with prior information. Solid lines represent quantum information (states), and double lines represent classical information. The player is in possession of a fixed quantum state $\rho$ and a POVM set $\mathbb{M}_{A|X} = \{M_{a|x}\}$. The referee has an instrument set $\Phi_{B|Y}=\{\phi_{b|y}(\cdot)\}$. The player sends state $\rho$ to the referee, who then implements one of the instruments $\{\phi_{b|y}(\cdot)\}$, according to the PMF $p_Y$. The output of the instrument is the measurement outcome $b$ and a post-measured state $\rho_{b|y}$. The referee sends both state $\rho_{b|y}$ and label $y$ (prior information) back to the player. The goal of the game is for the player to correctly guess the output $b$. The player then proceeds to use both $y$ and $\rho_{b|y}$, together with his set of measurements $\mathbb{M}_{A|X} = \{M_{a|x}\}$, in order to generate a guess index $g$. The player is allowed to do pre and post-processing \eqref{eq:CP}. The player then sends the guess $g$ to the referee, who then checks the winning condition ($g=b$).
    }
    \label{fig:fig2}
\end{figure}
	
\section{Advantage provided by resourceful states and POVM sets}
\label{s:s4}

We start by addressing quantum resource theories (QRTs) of states and sets of measurements with arbitrary resources. In order to do this, we need some elements from conic programming \cite{GM_2012, NJ_2012}. In particular, we need the concept of a closed convex cone (CCC). First, a set $\mathcal{K} \subseteq \LL_S(\hi)$ is a \emph{cone} if $\lambda K\in \mathcal{K}$, $\forall K \in \mathcal{K}$ and $\forall \lambda \geq 0$. Second, a cone $\mathcal{K}$ is \emph{convex} when $K_1+K_2\in \mathcal{K}$, $\forall K_1, K_2\in\mathcal{K}$. Third, a \emph{closed} set here is considered as being closed under some operator topology such as the trace norm. We note that all operator topologies are equivalent in the finite dimensional case \cite{Bratteli1987, Schaefer1999}. 
Fourth, the \emph{dual} of a cone $\mathcal{K}$ is the set defined as $\mathcal{K}^\circ \coloneqq \{O \in \LL_S(\hi)| \braket{O,K}_{\mathrm{HS}} \geq 0, \forall K\in \mathcal{K}\}$, where $\braket{C,D}_{\mathrm{HS}}=\Tr[C^\dagger D]$ is the Hilbert-Schmidt inner product in $\LL(\hi)$.
For any cone $\mathcal{K}$ we have $(\mathcal{K}^\circ)^\circ = \overline{{\it conv}(\mathcal{K})}$ (the closure of the convex hull), and so for any CCC $\mathcal{K}$ we have $(\mathcal{K}^\circ)^\circ = \mathcal{K}$. We now consider a property of quantum states defining a CCC as $\mathcal{F}(\hi) \subseteq \LL_S^+(\hi)$. We will address the set of \textit{free states} as the set ${\rm F}\equiv {\rm F}(\hi) \coloneqq \{\rho \in \mathcal{F}(\hi)\mid \tr[\rho] = 1\}$. We say a state $\rho \notin {\rm F}$ is a \emph{resourceful} state and \emph{free (resourceless)} otherwise. Desirable properties of quantum states regarded as resources include the likes of entanglement, coherence, amongst others \cite{RT_review}. We now similarly consider a property of measurement sets defining a CCC and use it to introduce a \textit{free set of POVM sets} and denote it as $\mathbb{F}$. We say a POVM set $\mathbb{M}_{A|X} \notin \mathbb{F}$ is a \emph{resourceful} POVM set, and a \emph{free (resourceless)} POVM set otherwise. Desirable properties of POVM sets include measurement incompatibility as a particular case \cite{RT_review}. We now address resource quantifiers. Consider a QRT of a quantum object $O$ being either a quantum state or a POVM set. The generalised robustness of resource and the weight of resource of  $O$ are given by:
\begin{align}
    {\rm R_F}
    \left(
    O
    \right)
    &\coloneqq
    \begin{matrix}
        \text{\small \rm min}\\
        r \geq 0\\
        O^F \in {\rm F} \\
        O^G \\
    \end{matrix}
    \left\{ 
    \rule{0cm}{0.6cm} r\,\bigg| \, 
    O
    +
    rO^G
    =
    (1+r)
    O^F
    \right\},
    \label{eq:RoRs}\\
    {\rm W_F}
    \left(
    O
    \right)
    &\coloneqq 
    \begin{matrix}
        \text{\small \rm min}\\
        w \geq 0\\
        O^F \in {\rm F} \\
        O^G \\
    \end{matrix}
    \left\{ 
    \rule{0cm}{0.6cm} w\,\bigg| \, 
    O
    =
    w\,
    O^G
    +(1-w)
    O^F
    \right\}.
    \label{eq:WoRs}
\end{align}
The generalised robustness quantifies the minimum amount of a general quantum object $O^G$ (either a general state $\rho^G$ or a general POVM set $\mathbb{M}^G_{A|X}$) that has to be added to the quantum object $O$ (either state $\rho$ or POVM set $\mathbb{M}_{A|X}$) to get a free object $O^F$ (either a free state $\rho^F$ or a free POVM set $\mathbb{M}^F_{A|X}$). The weight on the other hand, quantifies the minimum amount of a general object $O^G$ (either a general state $\rho^G$ or a general POVM set $\mathbb{M}^G_{A|X}$) that has to be used for recovering the quantum object $O$ (either state $\rho$ or POVM set $\mathbb{M}_{A|X}$) from a free object $O^F$. It is well known that these resource quantifiers can be written as conic programs \cite{RT_review, RT1, RT2, DS2020, uola}. We are now dealing with multiple objects, so it is natural to introduce the following notation. We say that a state-POVM set pair $(\rho,\mathbb{M}_{A|X})$ is: \emph{fully free} when both objects are free, \emph{partially resourceful} when either is resourceful, and \emph{fully resourceful} when both are resourceful. We will be addressing, from now on, QRTs of POVM sets for which CPreP and CPosP are free operations, meaning that the free set of POVM sets is closed under simulability of POVM sets. Having established these elements from conic programming (further details in \cref{a:pre}), we are now ready to analyse multi-object tasks from the point of view of QRTs.  

The main motivation now is to address the multi-object quantum subchannel discrimination games introduced in the previous section and compare the performance of a player having access to a potentially resourceful pair $(\rho, \mathbb{M}_{A|X})$, against the performance of a player having access only to free resources $(\sigma, \mathbb{N}_{A|X}) \in \mathrm{F} \times \mathds{F}$ (fully free player). We want the comparison to be fair, and so both players are compared when playing the same game (same ensemble of instruments $\{ p_Y,\Phi_{B|Y}\}$). We can then compare the performance of both players by analysing the following ratio:
\begin{align}
    \frac{
        P^{\rm D}_{\rm succ}(
        p_Y,
        \Phi_{B|Y},
        \rho,
        \mathbb{M}_{A|X}
        )
        }{
        \displaystyle
        \max_{
        \sigma \in {\rm F}
        }
        \max_{
            \mathbb{N}_{A|X} \in \mathbb{F}
        }
        P^{\rm D}_{\rm succ}(
        p_Y,
        \Phi_{B|Y},
        \sigma,
        \mathbb{N}_{A|X}
        )
        }
        .
\end{align}
If this ratio is larger than one, it naturally means that the pair $(\rho, \mathbb{M}_{A|X})$ offers an advantage over all fully free pairs, as it leads to larger probability of winning. It is then desirable to derive upper bounds for this ratio, and to explore how large the ratio can get to be, meaning maximising the ratio over all possible games, as this would represent the best case scenario for the pair $(\rho, \mathbb{M}_{A|X})$.

We now establish connections between robustness-based (weight-based) resource quantifiers for states and POVM sets and multi-object subchannel discrimination (exclusion) games with prior information. We proceed to establish a first result comparing the performance of a \emph{fully} or \textit{partially} resourceful pair against all \emph{fully free} pairs. 

\begin{result} \label{r:r1}
    Consider a player  with a quantum state $\rho$ and a POVM set $\mathbb{M}_{A|X}$. The advantage provided by these two objects when playing multi-object subchannel discrimination games with prior information $\{p_Y,\Phi_{B|Y}\}$ is:
    \begin{align}
        \max_{
        \{p_Y,\Phi_{B|Y}\}
        } 
        \frac{
            P^{\rm D}_{\rm succ}(
            p_Y,
                \Phi_{B|Y},
            \rho,
            \mathbb{M}_{A|X}
            )
        }{
            \displaystyle
            \max_{
                \substack{
                    \sigma \in {\rm F}\\
                    \mathbb{N}_{A|X} \in \mathbb{F}
                }
            }
            P^{\rm D}_{\rm succ}(
            p_Y,
                \Phi_{B|Y},
            \sigma,
            \mathbb{N}_{A|X}
            )
        }&
        \nonumber
        \\= 
        \Big [1+{\rm R_F}(\rho)\Big]
        \Big[1+{\rm R_\mathbb{F}}
        (&
        \mathbb{M}_{A|X}
        )
        \Big],
        \label{eq:r1}
    \end{align}
    with the maximisation over all quantum subchannel discrimination games and the generalised robustness of resource of state and of POVM sets.
\end{result}
\begin{proof}(Sketch) The full proof of this result is in \cref{a:r1}. We address here a sketch of the proof. The proof of the statement has two parts. The first part is to show that the right hand side of \eqref{eq:r1} constitutes an upper bound. This first part employs the primal conic program of the generalised robustness, and it follows from relatively straightforward arguments. The second part of the statement, proving that the upper bound is achievable, is a more involved endeavour, and so we describe it next. Given a pair 
$
(
\rho
,
\mathbb{M}_{A|X}
)
$, it is possible to use the dual conic programs of the generalised robustness so to extract positive semidefinite operators $Z^\rho$ and $\{Z^{\mathbb{M}_{A|X}}_{a|x} \equiv Z^{\mathbb{M}}_{a|x}\}$, $a\in \{1,...,l\}$, $x\in \{1,...,\kappa\}$, satisfying desirable optimality properties. Using these operators, we can define the following subchannel discrimination game. Fix a PMF $p_Y$, and consider the game $\{p_Y,\Psi_{B|Y}\}$, $b\in \{1,...,l+J\}$, $y\in \{1,...,\kappa\}$, as:
    \begin{align}
        &\Psi^{
            (\rho,
            \mathbb{M}_{A|X}, p_X, J)
        }_{b|y}
        (\eta)
        \nonumber
        \\
        &\coloneqq
        \bigg \{
        \begin{matrix}
            \alpha
            \tr[Z^\rho \eta]\,
            Z^{\mathbb{M}}_{b|y}\,
            p(y)^{-1}
            , 
            & \hspace{-1.2cm}
            b=1,...,l\\
            \frac{1}{J}
            [1-F_y(\eta)]
            \chi
            &  
            b=l+1,...,l+J
        \end{matrix}
    \end{align} 
    with $J\geq 1$ an integer, $\chi$ an arbitrary quantum state, $\alpha$ a coefficient, and $\{F_y(\cdot)\}$ functions that depend on $
    (
    \rho
    ,
    \mathbb{M}_{A|X}
    )
    $. These parameters are all specified in \cref{a:r1}. The next step is to analyse the performance of a player using fully free pairs to play this game. In \cref{a:r1} we derive the following upper bound respected by all fully free pairs 
    $
    (
    \sigma
    ,
    \mathbb{N}_{A|X}
    )
    $:
    \begin{align}
        &P^{\rm D}_{\rm succ}
        (
        p_Y,
        \Psi^{
            (
            \rho,
            \mathbb{M}_{A|X}, p_X,
            J
            )
        }_{B|Y}
        ,
        \sigma
        ,
        \mathbb{N}_{A|X}
        )
        \leq 
        \alpha 
        +
        \frac{
            1
        }{
            J
        }.
    \end{align}
    The next step is to analyse the performance of a player using now the fixed pair
    $
    (
    \rho
    ,
    \mathbb{M}_{A|X}
    )
    $. In \cref{a:r1} we also derive the lower bound:
    \begin{align}
        &P^{\rm D}_{\rm succ}
        (
        p_Y,
        \Psi^{(\rho
            ,
            \mathbb{M}_{A|X}, p_X,
            J
            )}_{B|Y},
        \rho
        ,
        \mathbb{M}_{A|X}
        )
        \nonumber
        \\
        &\geq
        \alpha 
        \Big[ 
        1+{\rm R_F}
        (\rho) 
        \Big]
        \Big[ 
        1+{\rm R_\mathbb{F}}
        (\mathbb{M}_{A|X}) 
        \Big].
    \end{align}
    Taking the previous two statements with $J\rightarrow \infty$ achieves the claim.
\end{proof}

We first note that this result applies to QRTs of states with \emph{arbitrary resources} and QRTs of POVM sets with \emph{arbitrary resources} for which POVM set simulability is a free operation and therefore, it covers as particular instances, several desirable resources for both states and POVM sets, such as measurement incompatibility. A second point to highlight is that the result holds for any pair $(\rho, \mathbb{M}_{A|X})$ that is either fully resourceful or even partially resourceful. Thirdly, this result generalises two sets of results from the literature:  i) the single-object results reported in \cite{paul_ivan_dani_2019, Teiko_2019, Uola_2019} as well as ii) the multi-object case of state-measurement pairs in \cite{MO1}. Let us address these cases in more detail.

First, the result in \cite{MO1} can immediately be recovered from \eqref{eq:r1} by restricting the POVM set to trivially have only one POVM, and similarly the operational task to consist on only one instrument from which we get:
\begin{align}
        \hspace{-0.1cm}
        \max_{\Phi_{}} 
        \hspace{-0.05cm}
        \frac{
            P^{\rm D}_{\rm succ}(
            \Phi_{},
            \rho,
            \mathbb{M}_{}
            )
        }{
            \displaystyle
            \max_{
                \substack{
                    \sigma \in {\rm F}\\
                    \mathbb{N}_{} \in \mathbb{F}
                }
            }
            P^{\rm D}_{\rm succ}(
            \Phi_{},
            \sigma,
            \mathbb{N}_{}
            )
        }
        \hspace{-0.1cm}
        = 
        \hspace{-0.1cm}
        \Big [1+{\rm R_F}(\rho)\Big]
        \Big[1+{\rm R_\mathbb{F}}
        (
        \mathbb{M}_{}
        ) 
        \Big],
\end{align}
thus explicitly recovering the result for state-measurement pairs reported in \cite{MO1}.

Second, let us now address how \eqref{eq:r1} also recovers the single-object results in \cite{paul_ivan_dani_2019, Teiko_2019, Uola_2019}. In \cite{paul_ivan_dani_2019, Teiko_2019, Uola_2019}, the player has access only to POVM sets (with no quantum states at their disposal), and the resource being exploited is specifically that of measurement incompatibility. Imposing these restrictions in our setup, our operational task reduces to that of \textit{single-object} quantum \textit{state} discrimination with prior information, and in turn the advantage is given by that of the generalised robustness of incompatibility alone. Let us see this explicitly. Consider the denominator in the ratio of interest:
\begin{align}
    \max_{
        \sigma \in {\rm F}
    }
    \max_{
        \mathbb{N}_{A|X} \in \mathbb{F}
    }
    P^{\rm D}_{\rm succ}(
    p_Y,
    \Phi_{B|Y},
    \sigma,
    \mathbb{N}_{A|X}
    ),
\end{align}
and the free state and POVM set achieving this as $(\sigma^*, \mathbb{N}_{A|X}^*)$. Consider now comparing the performance of a player using such fully free pair $(\sigma^*, \mathbb{N}_{A|X}^*)$ against a player using the partially resourceful pair $(\sigma^*, \mathbb{M}_{A|X})$. In this case, because the state is free, the right hand side in \eqref{eq:r1} is simply $[1+{\rm R}_\mathbb{F}(\mathbb{M}_{A|X})]$, and so we explicitly get: 
{\small\begin{align}
    \max_{\{p_Y, \Phi_{B|Y}\}} 
    \frac{
        P^{\rm D}_{\rm succ}(
        p_Y,
        \Phi_{B|Y},
        \sigma^*,
        \mathbb{M}_{A|X}
        )
    }{
        \displaystyle
        P^{\rm D}_{\rm succ}(
        p_Y,
        \Phi_{B|Y},
        \sigma^*,
        \mathbb{N}_{A|X}^*
        )
    }
    =\Big[1+{\rm R_\mathbb{F}}
    (
    \mathbb{M}_{A|X}
    ) 
    \Big].
\end{align}}
The left hand side of the latter expression can now be seen as an ensemble of states with prior information as $\mathcal{E}_{B|Y}\coloneqq \{\rho_{b|y}, p(y)\}$ with $\rho_{b|y}\coloneqq \phi_{b|y}(\sigma^*)$, and so we get:
\begin{align}
    \max_{\mathcal{E}_{B|Y}} 
    \frac{
        P^{\rm QSD-PI}_{\rm succ}(
        \mathcal{E}_{B|Y},
        \mathbb{M}_{A|X}
        )
    }{
        \displaystyle
        P^{\rm QSD-PI}_{\rm succ}(
        \mathcal{E}_{B|Y},
        \mathbb{N}_{A|X}^*
        )
    }
    = 
    \Big[1+{\rm R_\mathbb{F}}
    (
    \mathbb{M}_{A|X}
    ) 
    \Big],
\end{align}
thus explicitly recovering the results in \cite{paul_ivan_dani_2019, Teiko_2019, Uola_2019}. When considering subchannel discrimination games being played with a POVM set alone, the advantage we see becomes $[1+{\rm R}_\mathbb{F}(\mathbb{M}_{A|X})]$ \cite{paul_ivan_dani_2019, Teiko_2019, Uola_2019}. In the multi-object scenario considered here we instead get $[1+{\rm R_F}(\rho)][1+{\rm R}_\mathbb{F}(\mathbb{M}_{A|X})]$, which can be larger than $[1+{\rm R}_\mathbb{F}(\mathbb{M}_{A|X})]$, whenever $\rho$ is resourceful. This increment can be conceptually be understood by considering that, since we are now addressing a composite object, it is natural to expect each object to contribute to the overall advantage. Having said this however, it is still appealing that the advantage can be quantified in such a straightforward multiplicative manner.

We now proceed to show that Result 1 can also be extended to multi-object quantum subchannel \emph{exclusion} games with prior information, where it is now the \emph{weight} of resource the quantifier that characterises the advantage provided by resourceful pairs of states and POVM sets. In this scenario, since the figure of merit is now the probability of \emph{error} when playing the game, the ratio of interest is now:
\begin{align}
    \frac{
            P^{\rm E}_{\rm err}(
            p_Y,
            \Phi_{B|Y},
            \rho,
            \mathbb{M}_{A|X}
            )
        }{
            \displaystyle
        \min_{
        \sigma \in {\rm F}
        }
        \min_{
            \mathbb{N}_{A|X} \in \mathbb{F}
        }
        P^{\rm E}_{\rm err}(
        p_Y,
        \Phi_{B|Y},
        \sigma,
        \mathbb{N}_{A|X}
        )
        }
        .
\end{align}
If this ratio is smaller than one, it means the pair $(\rho, \mathbb{M}_{A|X})$ offers an advantage over all fully free pairs, as it leads to smaller probability of error. It is then desirable to derive lower bounds for this ratio, and to explore how small the ratio can get to be, meaning minimising the ratio over all possible games, as this would represent the best case scenario for the pair $(\rho, \mathbb{M}_{A|X})$.

\begin{result} \label{r:r2}
    Consider a player  with a quantum state $\rho$ and a POVM set $\mathbb{M}_{A|X}$, then, the advantage provided by these two objects when playing subchannel exclusion games with prior information $\{p_Y,\Phi_{B|Y}\}$ is given by:
    \begin{align}
        \min_{
        \{p_Y,\Phi_{B|Y}\}
        } 
        \frac{
            P^{\rm E}_{\rm err}(
            p_Y,
             \Phi_{B|Y},
            \rho,
            \mathbb{M}_{A|X}
            )
        }{
            \displaystyle
            \min_{
                \substack{
                    \sigma \in {\rm F}\\
                    \mathbb{N}_{A|X} \in \mathbb{F}
                }
            }
            P^{\rm E}_{\rm err}(
            p_Y,
            \Phi_{B|Y},
            \sigma,
            \mathbb{N}_{A|X}
            )
        }&
        \nonumber
        \\= 
        \Big [1-{\rm W_F}(\rho)\Big]
        \Big[1-{\rm W_\mathbb{F}}
        (&
        \mathbb{M}_{A|X}
        ) 
        \Big],
        \label{eq:r2}
    \end{align}
    with the minimisation over all quantum subchannel exclusion games and the weight of resource of state and of POVM sets.
\end{result}
The full proof of this result is in \cref{a:r2}. Similar to the discrimination case, this result also holds for arbitrary resources of both quantum states and POVM sets. It is illustrative to address a sketch of this proof, so in order to highlight the differences with the case for discrimination.

\begin{proof}(Sketch)
Similar to the discrimination case, that the right hand side of \eqref{eq:r2} constitutes a lower bound for the ratio of interest follows from the primal conic programs of the weight measures. Showing that such lower bound is achievable is a more involved endeavour, and so we describe it next. Given a pair 
    $
    (
    \rho
    ,
    \mathbb{M}_{A|X}
    )
    $, and the dual conic programs of the weight of resource, we can extract positive semidefinite operators $Y^\rho$ and $\{ Y^{\mathbb{M}_{A|X}}_{a|x} \equiv Y^{\mathbb{M}}_{a|x}\}$, $a\in \{1,...,l\}$, $x\in \{1,...,\kappa\}$, satisfying desirable optimality properties, so that we can construct the following  game. Fix a PMF $p_Y$, and consider the game $\{ p_Y,\Psi_{B|Y}\}$, $b\in \{1,...,l+1\}$, $y\in \{1,...,\kappa\}$, as:
\begin{align}
    &\Psi_{b|y}^{(\rho, \mathbb{M}_{A|X}, p_X)}(\eta) 
    \nonumber \\
    &\coloneqq
\begin{cases} 
    \beta \, 
    \tr[Y^{\rho} \eta]\,
    Y^{\mathbb{M}}_{b|y}\, 
    p(y)^{-1}, 
    & b = 1, \ldots, l \\
    [1 - G_y(\eta)]\,
    \xi^{\mathbb{M}_{A|X}}_y, & b = l + 1
\end{cases}
\end{align}
with $\{\xi^{\mathbb{M}_{A|X}}_y\}$ a set of quantum states, $\beta$ a coefficient, and $\{G_y(\cdot)\}$ functions that depend on $
(
\rho
,
\mathbb{M}_{A|X}
)
$. These parameters are all specified in \cref{a:r2}. We derive the following lower bound for all fully free pairs 
$
(
\sigma
,
\mathbb{N}_{A|X}
)
$:
\begin{align}
    \min_{\sigma \in F} \min_{\mathbb{N}_{A|X} \in \mathbb{F}} P^{\mathrm{E}}_{\text{ err}}(p_Y, \Psi^{(\rho, \mathbb{M}_{A|X}, p_X)}_{b|y} ,  \sigma, \mathbb{N}_{A|X})
    \geq
    \beta
    .
\end{align}
In \cref{a:r2} we also prove that for the fixed pair 
$
(
\rho
,
\mathbb{M}_{A|X}
)
$ we get the upper bound:
\begin{align}
    &P^{\mathrm{E}}_{\text{ err}}
    (
    p_Y 
    , 
    \Psi^{(\rho, \mathbb{M}_{A|X}, p_X)}_{b|y} 
    , 
    \rho
    , 
    \mathbb{M}_{A|X}
    ) 
    \nonumber
    \\
    &\leq
    \beta [1-{\rm W_F}(\rho)][1-{\rm W_F}(\mathbb{M}_{A|X})]
    .
\end{align}
These two statements together then achieve the claim.
\end{proof}
One technical point of comparison regarding the proofs of Result 1 \eqref{eq:r1} and Result 2 \eqref{eq:r2} is the nature of the game saturating the bound. In the discrimination case, the subchannel game needed to achieve the upper bound contains an infinite amount of extra subchannels $J \rightarrow \infty$, whilst in the exclusion case on the other hand, the subchannel game needed to achieve the lower bound requires only one extra subchannel. This difference can qualitatively be understood by taking into account that, the goal when maximising the ratio of interest is to make it difficult for the fully free players to perform well. In the discrimination case, this can be done by \emph{increasing} the amount of objects to discriminate from (as there will be more alternatives that are ``bad" options). In the exclusion case on the other hand, increasing the amount of objects makes it instead actually easier for the player in question to win (as there will be more alternatives that are ``good" options), and so in the exclusion case, in order to make it difficult for the fully free players to perform well, \emph{decreasing} the amount of objects to exclude from is desirable.
	
\section{Extension to general probabilistic theories}
\label{s:GPT}
Our main results \eqref{eq:r1} and \eqref{eq:r2} can be extended to general probabilistic theories (GPTs) \cite{Lami2017,Plavala2023,Takakura2022}. Let $V$ be a finite-dimensional Euclidean space and $V^*$ be its dual. We often identify $V^*$ with $V$ and the action $f(x)$ of $f\in V^*$ on $x\in V$ with the Euclidean inner product $\ang{f,x}$ of two vectors $f,x\in V$ (by means of the Riesz representation theorem \cite{Conway1985}). A \textit{GPT} is a pair of sets $(\Omega, \mathscr{E})$ such that $\Omega$ is a compact convex subset of $V$ with its affine hull satisfying $\mathit{aff}(\Omega)\not\owns 0$ and linear hull satisfying $\mathit{lin}(\Omega)=V$, and $\mathscr{E}=\{e\in V^*\mid e(\forall \omega)\in[0,1]\}$. The sets $\Omega$ and $\eff$ are called the \textit{state space} and \textit{effect space}, and their elements are called \textit{states} and \textit{effects} respectively. We remark that in this paper we assume the \textit{no-restriction hypothesis} \cite{PhysRevA.81.062348}.
Clearly, GPTs are generalisations of quantum theory: taking $\LL_S(\hi)$ as $V$ and $\mathcal{D}(\hi)$ as $\Omega$, we can recover the description of quantum states. It is also easy to see that effects are generalisation of POVM elements, and the Euclidean inner product $\ang{e,\omega}=e(\omega)~(\omega\in\Omega, e\in\eff)$ generalises the Hilbert-Schmidt inner product $\ang{M,\rho}_{\mathrm{HS}}=\Tr[M\rho]~(\rho\in\mathcal{D}(\hi), M\in\LL_S^{+}(\hi))$ for quantum theory. In particular, the \textit{unit effect} $u\in\eff$ defined via $u(\omega)=1 (\forall\omega\in\Omega)$ corresponds to the identity operator $\mathds{1}$. With similar notations to the quantum case, we define a \textit{measurement} with an outcome set $A=\{1,\ldots,l\}$ as a set of effects $\{e_a\}_{a=1}^l$ such that $\sum_{a\in A}e_a=u$. A \textit{measurement set} $\mathbb{E}_{A|X}=\{e_{a|x}\}_{a,x}$ is also introduced as a generalisation of a POVM set. The notion of (in)compatibility and simulability for POVMs can be extended naturally to measurements in GPTs by rephrasing \eqref{eq:comp} and \eqref{eq:CP} in terms of effects instead of POVM elements. It allows us to use the expression $\mathbb{N}_{B|Y} \preceq \mathbb{E}_{A|X}$ for two measurement sets $\mathbb{N}_{B|Y}$ and $\mathbb{E}_{A|X}$ in GPTs meaning that $\mathbb{N}_{B|Y}$ is simulable by $\mathbb{E}_{A|X}$. It is also clear that measurement compatibility is closed under simulation in GPTs.

For a GPT $(\Omega, \mathscr{E})$ whose underlying vector space is $V$, we define the \textit{positive cone} $V_+\subset V$ by $V_+=\mathit{cone}(\Omega)$ and the \textit{dual cone} $V_+^\circ\subset V^*$ by $V_+^\circ=\mathit{cone}(\mathscr{E})=\{f\in V^*\mid f(x)\ge0, \forall x\in V_+\}$, where $\mathit{cone}(\cdot)$ is the conic hull of the set.
For the quantum case, these cones correspond to the set of positive semi-definite operators $\LL_{S}^+(\hi)$. A linear map $\xi\colon V\to V$ is called \textit{positive} if $\xi(V_+)\subset V_+$.
A positive map $\xi$ is called a \textit{channel} if it satisfies $\ang{u,\xi(\omega)}=1,\forall\omega\in\Omega$, and a \textit{subchannel} if $\ang{u,\xi(\omega)}\le1,\forall\omega\in\Omega$.
An \textit{instrument} $\Xi$ is defined as a set of subchannels $\Xi=\{\xi_b\}_{b=1}^m$ such that $\sum_{b=1}^m\xi_b$ is a channel, and we also introduce an \textit{instrument set} $\Xi_{B|Y}=\{\Xi_y\}_{y=1}^\tau=\{\xi_{b|y}\}_{b,y}$ in the same way as the quantum case.
We remark that here we do not require complete positivity for (sub)channels due to the difficulty in determining composite systems for GPTs \cite{Aubrun2021}.
It seems that quantum description for channels is not recovered, but such inconsistency can be overcome by restricting ourselves to a (convex) subset of the set of all channels in GPTs.

Now we are in position to present our third main result.
We consider \textit{general-probabilistic} multi-object subchannel discrimination game with prior information (GPScD-PI) as natural extension of QScD-PI to GPTs. 
The scenario of GPScD-PI is the same as QScD-PI: replacing a quantum state $\rho\in\mathcal{D}(\hi)$, a POVM set $\mathbb{M}_{A|X}=\{M_{a|x}\}$, and a quantum instrument set $\Phi_{B|Y}$ with a state $\omega\in \Omega$, a measurement set $\mathbb{E}_{A|X}=\{e_{a|x}\}$, and an instrument set $\Xi_{B|Y}$ for a GPT $(\Omega, \eff)$. 
The maximum success probability is given in the same way as \eqref{eq:Psuc_Q} by
\begin{align}
    P^{\rm GPD}_{\rm succ}
    (
    &
    p_Y,
    \Xi_{B|Y},
    \omega,
    \mathbb{E}_{A|X}
    )
    \coloneqq
    \max_{\mathcal{S}}
    \hspace{-0.2cm}
    \sum_{g,a,x,\mu,b,y}
    \hspace{-0.3cm}
    \delta_{g,b}\,
    s(g|a,y,z)
    \nonumber
    \\
    &
    \times
    \ang{
        e_{a|x},\,
        \xi_{b|y}
        (\omega)
    }
    r(x|y,z)\,
    q(z)\,
    p(y)\label{eq:Psuc_G}
\end{align}
with the maximisation over all possible strategies $\mathcal{S} = \{q_Z, r_{X|YZ}, s_{G|AYZ}\}$. Similarly, we can consider general-probabilistic multi-object subchannel exclusion game with prior information (GPScE-PI)  generalising the QScE-PI and \eqref{eq:Perr_Q}, and the minimum error probability is 
\begin{align}
    P^{\rm GPE}_{\rm err}
    (
    &
    p_Y,
    \Xi_{B|Y},
    \omega,
    \mathbb{E}_{A|X}
    )
    \coloneqq
    \min_{\mathcal{S}}
    \hspace{-0.2cm}
    \sum_{g,a,x,\mu,b,y}
    \hspace{-0.3cm}
    \delta_{g,b}\,
    s(g|a,y,z)
    \nonumber
    \\
    &
    \times
    \ang{
        e_{a|x},\,
        \xi_{b|y}
        (\omega)
    }
    r(x|y,z)\,
    q(z)\,
    p(y)\label{eq:Perr_G}
\end{align}
with the minimisation over all possible strategies $\mathcal{S} = \{q_Z, r_{X|YZ}, s_{G|AYZ}\}$.
The last step we need is to introduce the notion of resourceful sets in GPTs, and this can be also done by generalising quantum concepts straightforwardly.
In fact, with a CCC $\mathcal{F}\subseteq V_+$ (closed with respect to e.g. the Euclidean topology in $V$),  we can introduce free states as elements of the set $\mathrm{F}=\{\omega\in \mathcal{F}\mid\ang{u,\omega}=1\}$ and resourceful states as $\Omega\setminus \mathrm{F}$.
A set $\mathbb{F}$ of free measurement sets is similarly extended, and $\mathbb{M}_{A|X}\notin\mathbb{F}$ is called resourceful.
We can also introduce the generalised robustness of resource and the weight of resource of an object $O$ (either a state or a measurement set) in the GPT $(\Omega, \eff)$ as
{\small
    \begin{align}
        {\rm R_F^{\mathrm{GP}}}
        (
        O
        )
        \hspace{-0.5mm}&\coloneqq
        \begin{matrix}
            \text{\small \rm min}\\
            r \geq 0\\
            O^F \in {\rm F} \\
            O^G \\
        \end{matrix}
        \left\{ 
        \rule{0cm}{0.6cm} r\,\bigg| \, 
        O
        +
        rO^G
        =
        (1+r)
        O^F
        \right\},
        \label{eq:RoRs_GPT}\\
        {\rm W^{\mathrm{GP}}_F}
        (
        O
        )
        \hspace{-0.5mm}&\coloneqq 
        \begin{matrix}
            \text{\small \rm min}\\
            w \geq 0\\
            O^F \in {\rm F} \\
            O^G \\
        \end{matrix}
        \left\{ 
        \rule{0cm}{0.6cm} w\,\bigg| 
        O
        =
        w\,
        O^G
        +(1-w)
        O^F\hspace{-0.5mm}
        \right\}.
        \label{eq:WoRs_GPT}
\end{align}}
These quantities are associated with conic programs with the positive cone $V_+$ and dual cones $V_+^\circ$ generated respectively by $\Omega$ and $\eff$.
Now we can generalise the results from the quantum domain \eqref{eq:r1} and \eqref{eq:r2} to GPTs:
\begin{result} \label{r:r3}
    Consider a player with a state $\omega$ and a measurement set $\mathbb{E}_{A|X}$ of a GPT $(\Omega,\eff)$. 
    The advantage provided by these two objects when playing general-probabilistic multi-object subchannel discrimination and exclusion games with prior information $\{p_Y,\Xi_{B|Y}\}$ is:
    \begin{align}
        \hspace{-0.1cm}
        \max_{\{p_Y,\Xi_{B|Y}\}} 
        \hspace{-0.05cm}
        \frac{
            P^{\rm GPD}_{\rm succ}(
            p_Y,
            \Xi_{B|Y},
            \omega,
            \mathbb{E}_{A|X}
            )
        }{
            \displaystyle
            \max_{
                \substack{
                    \sigma \in {\rm F}\\
                    \mathbb{N}_{A|X} \in \mathbb{F}
                }
            }
            P^{\rm GPD}_{\rm succ}(
            p_Y,
            \Xi_{B|Y},
            \sigma,
            \mathbb{N}_{A|X}
            )
        }&
        \nonumber
        \\= 
        \Big [1+{\rm R_F^{GP}}(\omega)\Big]
        \Big[1+{\rm R^{GP}_\mathbb{F}}&
        (
        \mathbb{E}_{A|X}
        ) 
        \Big],
        \label{eq:result3-1}
    \end{align}

    and 
    
    \begin{align}
        \hspace{-0.1cm}\min_{\{p_Y,\Xi_{B|Y}\}}
        \hspace{-0.05cm}
        \frac{
            P^{\rm GPE}_{\rm err}(
            p_Y,
            \Xi_{B|Y},
            \omega,
            \mathbb{E}_{A|X}
            )
        }{
            \displaystyle
            \min_{
                \substack{
                    \sigma \in {\rm F}\\
                    \mathbb{N}_{A|X} \in \mathbb{F}
                }
            }
            P^{\rm GPE}_{\rm err}(
            p_Y,
            \Xi_{B|Y},
            \sigma,
            \mathbb{N}_{A|X}
            )
        }&
        \nonumber
        \\= 
        \Big [1-{\rm W^{GP}_F}(\omega)\Big]
        \Big[1-{\rm W^{GP}_\mathbb{F}}&
        (
        \mathbb{E}_{A|X}
        ) 
        \Big],
        \label{eq:result3-2}
    \end{align}
    respectively. With the maximisation (minimisation) over all general-probabilistic subchannel discrimination (exclusion) games and the generalised robustness (weight) of resource of state and of measurement sets.
\end{result}
The full proof of this result is in \cref{app:GPT}. These proofs proceed in a similar way as the previous ones by appropriately extending the quantum objects to their counterparts in GPTs.
Some examples of this include using the Euclidean inner product and the \textit{order unit norm} instead of the Hilbert-Schmidt inner product and the operator norm respectively.
	
\section{Conclusions}

In this work we have introduced \textit{multi-object} operational tasks for measurement incompatibility in the form of subchannel discrimination and exclusion games with prior information, where the player can simultaneously harness the resources contained within two quantum objects; a quantum state, and a set of measurements (POVM set). Specifically, we have shown that \emph{any} fully or partially resourceful pair (state, POVM set) is useful for a suitably chosen multi-object subchannel discrimination and exclusion game with prior information. We have found that, when compared to the best possible strategy using \emph{fully free} state-POVM set pairs, the advantage provided by a pair state-POVM set can be quantified, in a \emph{multiplicative} manner, by the resource quantifiers of generalised robustness and weight of resource of the state and the POVM set, for discrimination and exclusion respectively. These results hold true for arbitrary resources of quantum states as well as arbitrary resources of POVM sets closed under classical pre and post-processing. The results presented here are therefore telling us that all sets of incompatible measurement can be useful for multi-object operational tasks. These results also happen to generalise various other previous results in the literature. First, it generalises the results reported in \cite{paul_ivan_dani_2019, Teiko_2019, Uola_2019}, where \textit{single-object} operational tasks for measurement incompatibility were characterised. Second, it also generalises the operational tasks for state-measurement pairs introduced in \cite{MO1}. Third and finally, we furthermore have shown that these results can be extended to the realm of general probabilistic theories and, consequently, also generalise the \textit{single-object} results for GPTs reported in \cite{RT2}, now to the \textit{multi-object} regime.

There are various different directions where to further explored these findings. First, quantum resource theories have recently been explored beyond the realm of convexity and so, it would be interesting to explore multi-object tasks in such regimes \cite{bc1, bc2, bc3, MR2}. Second, whilst we have explored both discrimination and exclusion games, it is also known that these games can be considered more generally as quantum state betting games \cite{EUT_QSB, EUT_QSB2},  from the point of view of expected utility theory, and so it would be interesting to explore such betting games using the multi-object perspective employed in this paper. Third, it could also be interesting to explore these findings from the point of view of cooperative game theory, akin to the setting being explored for multi-resource tasks in \cite{MR2}. Fourth, it could also be relevant to explore multi-object tasks for more general scenarios involving the incompatibility of sets of instruments \cite{instru0, instru1, instru2, instru3, instru4, instru5, instru6, instru7, instru8}. 
	
\emph{Acknowledgements.---}A.F.D. thanks Paul Skrzypczyk for insightful discussions. A.F.D. and R.T. thank Ivan \v{S}upi\'{c} and Erkka Haapasalo for interesting discussions on measurement incompatibility. A.F.D. acknowledges support from the International Research Unit of Quantum Information, Kyoto University, the Center for Gravitational Physics and Quantum Information (CGPQI), and COLCIENCIAS 756-2016. 
R.T. acknowledges support from MEXT QLEAP and JST COI-NEXT program Grant No. JPMJPF2014.
C.E.S. and F.J.H. acknowledges support from University of C\'ordoba (Grants FCB-12-23). 

\newpage
\bibliographystyle{apsrev4-1}
\bibliography{bibliography2.bib}

\onecolumngrid
\appendix
	
\section{Preliminaries}
\label{a:pre}
Consider a set of POVMs $\mathbb{M}_{A|X} \coloneqq \{M_{a|x}\}$, $M_{a|x}\geq 0$, $\forall a,x$, $\sum_a M_{a|x} = \mathds{1}$, $\forall x$, $x=1,...,\kappa$, with each POVM having elements $a=1,...,l$, a set of subchannels $\Phi_{B|Y} \coloneqq \{\Phi_{b|y}(\cdot)\}$, $\Phi_{b|y}(\cdot)$ a completely positive (CP) trace-nonincreasing map $\forall b,y$, $\sum_b \Phi_{b|y}(\cdot) = \Phi_y (\cdot)$ a trace-preserving (TP) map $\forall y$, $y=1,...,m$, $b=1,...,n$, and a quantum state $\rho$, $\rho \geq 0$, $\tr(\rho) = 1$. We start by rewriting the probability of success in quantum subchannel discrimination (QScD) with prior information \cite{Teiko_2018} as:
\begin{align}
    P^{\rm D}_{\rm succ}(
    p_Y,
    \Phi_{B|Y}, 
    \rho,
    \mathbb{M}_{A|X}
    )
    &\coloneqq
    \max_{\mathcal{S}}
    \sum_{g,a,x,z,b,y}
    \delta_{g,b}\,
    s(g|a,y,z)
    \tr[
    M_{a|x}
    \Phi_{b|y}
    (\rho)
    ]\,
    r(x|y,z)\,
    q(z)\,
    p(y)
    \\
    &=
    \max_{\mathcal{S}}
    \sum_{a,x,z,b,y}
    s(b|a,y,z)
    \tr[
    M_{a|x}
    \Phi_{b|y}
    (\rho)
    ]\,
    r(x|y,z)\,
    q(z)\,
    p(y)
    \\
    &=
    \max_{\mathcal{S}}
    \sum_{b,y}
    \tr
    \left[
    \left(
    \sum_{a,x,z}
    s(b|a,y,z)\,
    M_{a|x}\,
    r(x|y,z)\,
    q(z)
    \right)
    \Phi_{b|y}
    (\rho)
    \right]
    p(y)
    \\
    &=\max_{
        \mathbb{N}_{B|Y}
        \preceq \,
        \mathbb{M}_{A|X}
    }
    \sum_{b,y}
    \tr
    \left[
    N_{b|y}
    \Phi_{b|y}
    (\rho)
    \right] 
    p(y)
    ,\label{eq:app_sim}
\end{align}
with the set of strategies 
$\mathcal{S} = \{q_Z, p_{X|YZ}, p_{B|AYZ}\}$ and the simulability of POVM sets
$
\mathbb{N}_{B|Y}
\preceq \,
\mathbb{M}_{A|X}
$ defined as:
\begin{align}
    N_{b|y}
    =
    \sum_{a,x,z}
    s(b|a,y,z)\,
    M_{a|x}\,
    r(x|y,z)\,
    q(z)
    .
\end{align}
\begin{lemma}
    \label{lem}
    (Dual conic programs for the generalised robustness of state and POVM sets) The generalised robustness of resource of a state $\rho$ and a POVM set $\mathbb{M}=\{M_{a|x}\}$, $x\in\{1,...,\kappa\}$, $a\in\{1,...,l\}$ is given by:
    \begin{subequations}
        \begin{align}
            {\rm R_F}\left(\rho\right)=
            \max_{Z} \hspace{0.2cm}
            &\tr[Z\rho]-1, \label{eq:RoRs1}\\
            {\rm s. t.} \hspace{0.2cm}
            & Z\geq 0,\label{eq:RoRs2}\\
            &\tr[Z\sigma]\leq 1, \hspace{0.4cm} \forall \sigma \in {\rm F} ,
            \label{eq:RoRs3}
        \end{align}
    \end{subequations}
    and
    \begin{subequations}
        \begin{align}
            {\rm R_\mathbb{F}}
            \left(
            \mathbb{M}_{A|X}
            \right)
            =
            \max_{\{Z_{a,x}\}} 
            \hspace{0.2cm}
            &
            \sum_{x=1}^{\kappa}
            \sum_{a=1}^{l}
            \tr[
            Z_{a,x}
            M_{a|x}
            ]
            -1, 
            \label{eq:RoRm1}\\
            {\rm s. t.} 
            \hspace{0.2cm}
            & Z_{a,x}\geq 0, 
            \hspace{0.5cm}
            \forall a, x, 
            \label{eq:RoRm2}\\
            &
            \sum_{x=1}^{\kappa}
            \sum_{a=1}^{l}
            \tr[
            Z_{a,x} N_{a|x}
            ]
            \leq
            1, 
            \hspace{0.2cm}
            \forall \mathbb{N}_{A|X}=\{N_{a|x}\} \in \mathbb{F}.
            \label{eq:RoRm3}
        \end{align}
    \end{subequations}
    These are the dual conic formulations of the generalised robustnesses for states and POVM sets, respectively.
\end{lemma}
	
\begin{proof}
    The dual version of the generalised robustness of resource of a state is a known and well reported  case in the literature (see for instance \cite{WoE_Brandao, RT1}), and so we only address here the case for POVM sets. This proof follows similar techniques to that of states, and we present it below for completeness. 
    For simplicity, we use a symbol $\overset{x,a}{\oplus} K$ to represent the direct sum of the identical subset (or element) $K$ of the linear space $\LL_{S}(\hi)$, i.e., $\overset{x,a}{\oplus} K=\oplus_{x=1}^\kappa\oplus_{a=1}^l K=K\oplus K\oplus \cdots\oplus K$.
    A POVM set $\mathbb{M}_{A|X}=\{M_{a|x}\}$ ($x\in\{1,...,\kappa\}$, $a\in\{1,...,l\}$) is naturally an element of the direct sum $\overset{x,a}{\oplus} \LL_{S}(\hi)$, and the generalised robustness of resource of $\mathbb{M}=\{M_{a|x}\}$ is given as
    \begin{align}
            {\rm R_\mathbb{F}}
        \left(
        \mathbb{M}_{A|X}
        \right)
        &\coloneqq
        \begin{matrix}
            \text{\small \rm min}\\
            r \geq 0\\
            \mathbb{M}_{A|X}^F \in {\rm \mathbb{F}} \\
            \mathbb{M}_{A|X}^G \\
        \end{matrix}
        \left\{ 
        \rule{0cm}{0.6cm} r\,\bigg| \, 
        \mathbb{M}_{A|X}+r \mathbb{M}_{A|X}^G=(1+r)\mathbb{M}_{A|X}^F
        \right\}.
        \label{rob}
    \end{align}
    Let us choose an arbitrary $\rho_0\in\mathcal{D}(\hi)$ to construct $\overset{x,a}{\oplus}\rho_0\in \overset{x,a}{\oplus}\LL(\hi)$.
     The direct sum $\overset{x,a}{\oplus}\LL(\hi)$ equips an inner product $\ang{\cdot,\cdot}$ induced naturally by the Hilbert-Schmidt inner product $\ang{\cdot,\cdot}_{\mathrm{HS}}$ on $\LL(\hi)$.
    Since $\ang{\mathbb{M}_{A|X}^F,\overset{x,a}{\oplus}\rho_0}=\sum_{x,a}\ang{M^F_{a|x},\rho_0}_{\mathrm{HS}}=\kappa$ holds, 
    we can rewrite \eqref{rob} as a conic program with primal variables $O_{A|X}:=(1+r)\mathbb{M}_{A|X}^F$:
    \begin{align}
        1
        +
        {\rm R_\mathbb{F}}
        \left(
        \mathbb{M}_{A|X}
        \right)
        =
        \min_{O_{A|X}} \hspace{0.2cm}
        &	\ang{O_{A|X},\frac{1}{\kappa}\overset{x,a}{\oplus}\rho_0}
        \\
        {\rm s. t.} \hspace{0.2cm}
        & \mathbb{M}_{A|X}\leq_{\overset{x,a}{\oplus}\LL_{S}^+} O_{A|X},\\
        &
        O_{A|X} \in \mathfrak{F}(\hi),
    \end{align}
    where we defined $\mathbb{M}_{A|X}^F\ge_{\overset{x,a}{\oplus}\LL_S^+}O_{A|X}$ by $O_{A|X}-\mathbb{M}_{A|X}\in \overset{x,a}{\oplus}\LL_{S}^+$ and $\mathfrak{F}(\hi):=\mathit{cone}(\mathbb{F})$. Expressing $O_{A|X}=\{O_{a|x}\}$, we explicitly have
    \begin{align}
        1
        +
        {\rm R_\mathbb{F}}
        \left(
        \mathbb{M}_{A|X}
        \right)
        =
        \min_{O_{A|X}} \hspace{0.2cm}
        &\frac{1}{\kappa}
        \sum_{a,x}
        \tr
        \left[
        O_{a|x}\rho_0
        \right], 
        \label{eq:RoRMS1}\\
        {\rm s. t.} \hspace{0.2cm}
        & M_{a|x}\leq O_{a|x}, \hspace{0.2cm} \forall a,x,\label{eq:RoRMS2}\\
        &O_{A|X} \in \mathfrak{F}(\hi).
    \end{align}
    The primal constraints can alternatively be written as i) $O_{a|x}-M_{a|x} \geq 0$, $\forall a,x$, and ii) $\tr[O_{a|x}Q_{a|x}] \geq 0$, $\forall a,x$, $\forall Q_{A|X}\in \mathfrak{F}(\hi)^\circ$. Consider now a set of dual variables $Z_{A|X} \coloneqq \{Z_{a|x}\}\in\overset{x,a}{\oplus}\LL(\hi)$ with a first dual constraint i) $Z_{A|X} \geq_{\overset{x,a}{\oplus}\LL_S^+(\hi)} 0$. Similarly, for the second primal constraints consider $g(Q_{A|X})\geq 0$, $\forall Q_{A|X} \in \mathfrak{F}(\hi)^\circ$. Let us now construct the Lagrangian:
    \begin{align}
        \mathcal{L}
        \left(
        O_{A|X}
        ,
        Z_{A|X}
        ,
        \mathbb{M}_{A|X}
        \right)
        &\coloneqq
        \frac{1}{\kappa}
        \ang{O_{A|X},\overset{x,a}{\oplus}\rho_0}
        -
        \ang{
            \left(
            O_{A|X}
            -
            \mathbb{M}_{A|X}
            \right),
            Z_{A|X}
        }
        -
        \int_{
            Q_{A|X}\in \mathfrak{F}(\hi)^\circ
        }
        \hspace{-1.2cm}
        dQ\,
        g(Q_{A|X})
        \ang
        {
            O_{A|X},
            Q_{A|X}
        }
        \\	
        &=
        \ang
        {
        \mathbb{M}_{A|X},~
        Z_{A|X}
    }   
        +
        \ang
    {
        O_{A|X},~
        \left(
        \frac{1}{\kappa}\overset{x,a}{\oplus}\rho_0
        -
        Z_{A|X}
        -
        Q_{A|X}'
        \right)
    }
        ,
        \hspace{0.1cm}
        Q_{A|X}'
        \coloneqq
        \hspace{-0.2cm}
        \int_{
            Q_{A|X}\in \mathfrak{F}(\hi)^\circ
        }
        \hspace{-1.2cm}
        dQ\,
        g(Q_{A|X})
        Q_{A|X}\\
        &\left(=
        \sum_{a,x}
        \tr
        \left[
        M_{a|x}
        Z_{a|x}
        \right]   
        +
        \sum_{a,x}
        \tr
        \left[ 
        O_{a|x}
        \left(
        \frac{1}{\kappa}\rho_0
        -
        Z_{a|x}
        -
        Q_{a|x}'
        \right)
        \right]\right)
        .
    \end{align} 
We have that $Q_{A|X}'\coloneqq\{Q_{a|x}'\} \in \mathfrak{F}(\hi)^\circ$ because $Q_{A|X} \in \mathfrak{F}(\hi)^\circ$ and $\mathfrak{F}(\hi)^\circ$ is a convex cone. By construction, the Lagrangian satisfies
$
\mathcal{L}
(
O_{A|X}
,
Z_{A|X}
,
\mathbb{M}_{A|X}
)
\leq
1
+
{\rm R_\mathbb{F}}
\left(
\mathbb{M}_{A|X}
\right)
$. We can now eliminate the Lagrangian dependence on the primal variables by imposing suitable constraints on the dual variables as ii)
$
\frac{1}{\kappa}\overset{x,a}{\oplus}\rho_0
-
Z_{A|X}
-
Q_{A|X}'=0
$
(i.e., $\frac{1}{\kappa}\rho_0-Z_{a|x}
-
Q_{a|x}'=0$, $\forall a,x$)
We can then multiply these second dual constraints with $\mathbb{N}_{A|X}=\{N_{a|x}\} \in \mathbb{F}$ to get
$
1=
\ang{
Z_{A|X},
N_{A|X}
}
+
\langle
Q_{A|X}',
N_{A|X}
\rangle
$. The last term in the r.h.s is non-negative (since $Q_{A|X}' \in \mathfrak{F}(\hi)^\circ$ and $N_{A|X}\in \mathfrak{F}(\hi)$) and so we get
$
1
\geq
\ang{
Z_{A|X}, 
N_{A|X}
}
$. 
Maximising the Lagrangian over the dual variables subject to the dual constraints achieves the upper bound because of strong duality. Strong duality in turn follows from Slatter's condition, since there exists a strictly feasible choice for $Z_{A|X}$, take for instance $Z_{a|x} = \frac{\mathds{1}}{2d\kappa}$, $\forall a,x$. This choice satisfies i) $Z_{a|x}>0$, $\forall a,x$, and ii) $\ang{Z_{A|X},N_{A|X}}=\sum_{a,x} \tr[Z_{a|x}N_{a|x}] = \frac{1}{d\kappa} \sum_{a,x} \tr[N_{a|x}] = \frac{1}{2d\kappa} \sum_{x} \tr[\mathds{1}] =\frac{1}{2} <1$. Overall, the dual conic program reads:
\begin{align*}
    1
    +
    {\rm R_\mathbb{F}}
    \left(
    \mathbb{M}_{A|X}
    \right)
    =
    \max_{Z_{A|X}} \hspace{0.2cm}
    &
    \ang{
    M_{a|x},
    Z_{a|x}
},\\
    {\rm s. t.} \hspace{0.2cm}
    & Z_{A|X}\geq_{\overset{x,a}{\oplus}\LL_{S}^+} 0, \hspace{0.2cm} \\
    &
    \ang{
    Z_{A|X}, 
    \mathbb{N}_{A|X}
}
    \leq 1,
    \hspace{0.2cm}
    \forall \mathbb{N}_{A|X} \in \mathbb{F},
\end{align*}
or
\begin{align*}
    1
    +
    {\rm R_\mathbb{F}}
    \left(
    \mathbb{M}_{A|X}
    \right)
    =
    \max_{Z_{A|X}} \hspace{0.2cm}
    &
    \sum_{a,x}
    \tr
    \left[
    M_{a|x}
    Z_{a|x}
    \right],\\
    {\rm s. t.} \hspace{0.2cm}
    & Z_{a|x}\geq 0, \hspace{0.2cm} \forall a,x,\\
    &
    \sum_{a,x}
    \tr
    \left[ 
    Z_{a|x} 
    N_{a|x}
    \right]
    \leq 1,
    \hspace{0.2cm}
    \forall \mathbb{N}_{A|X} \in \mathbb{F},
\end{align*}
and thus achieving the claim.	
\end{proof}

Lemma \ref{lem} can be also rewritten for the exclusion scenario. In short, it reads:

\begin{lemma} 
\label{lem:2}
(Dual conic programs for the weight of resource of states and POVM sets) The weight of resource of a state $\rho$ and a POVM set $\mathbb{M} = \{M_{a|x}\}$ for $x \in \{1, \ldots, \kappa\}, a \in \{1, \ldots, l\}$ can be written as:
	\begin{equation}\label{eq:WoRs3}
		\begin{aligned}
			{\rm W_F}(\rho) = \max_Y
			\hspace{0.2cm} 
			&\tr[(-Y)\rho] + 1, \\
			{\rm s. t.} \hspace{0.2cm}
			& Y\geq 0,\\
			&\tr[Y\sigma]\geq 1, \hspace{0.4cm} \forall \sigma \in {\rm F} ,
		\end{aligned}
	\end{equation}

    and
\begin{equation}\label{eq:WoRm3}
	\begin{aligned}
	{\rm W_{F}}(\mathbb{M}_{a|x}) = \max_{\{Y_{a,x}\}}\hspace{0.2cm} &\sum_{x=1}^{\kappa} \sum_{a=1}^{l} \tr[(-Y_{a,x}) M_{a|x}] + 1,
		\\
		{\rm s. t.} 
		\hspace{0.2cm}
		& Y_{a,x}\geq 0, 
		\hspace{0.5cm}
		\forall a, x, 
		\\
		&
		\sum_{x=1}^{\kappa}
		\sum_{a=1}^{l}
		\tr[
		Y_{a,x} N_{a|x}
		]
		\geq
		1, 
		\hspace{0.2cm}
		\forall \mathbb{N}_{A|X}=\{N_{a|x}\} \in \mathbb{F},
	\end{aligned}
\end{equation}
respectively. 
\end{lemma}

\section{Proof of \cref{r:r1}}
\label{a:r1}
	
Consider a set of free states $F$ and a set of free POVM assemblages $\mathbb{F}$. The statement we want to prove is:
\begin{align}
    \hspace{-0.1cm}
    \max_{
    \{p_Y,\Phi_{B|Y}\}
    } 
    \frac{
        P^{\rm D}_{\rm succ}(
        p_Y,
        \Phi_{B|Y}, 
        \rho,
        \mathbb{M}_{A|X}
        )
    }{
        \displaystyle
        \max_{
            \substack{
                \sigma \in {\rm F}\\
                \mathbb{N}_{A|X} \in \mathbb{F}
            }
        }
        P^{\rm D}_{\rm succ}(
        p_Y,
        \Phi_{B|Y}, 
        \sigma,
        \mathbb{N}_{A|X}
        )
    }
    = \hspace{-0.1cm}
    \Big [1+{\rm R_F}(\rho)\Big]
    \Big[1+{\rm R_\mathbb{F}}
    (
    \mathbb{M}_{A|X}
    ) \Big],
    \label{eq:r1a}
\end{align}
with the maximisation over all sets of subchannels. We start by proving the upper bond.
\begin{proof}(Upper bound)
    Given \emph{any} game $(p_Y,\Phi_{B|Y})$ and \emph{any} pair $(\rho, \mathbb{M}_{A|X})$ we have:
    \begin{align}
        &P^{\rm D}_{\rm succ}
        (
        p_Y,
        \Phi_{B|Y}, 
        \rho,
        \mathbb{M}_{A|X}
        )
        \nonumber
        \\
        &
        =
        \max_{
            \mathbb{N}_{B|Y}
            \preceq \,
            \mathbb{M}_{A|X}
        }
        \sum_{b,y}
        \tr[
        N_{b|y}
        \Phi_{b|y}
        (\rho)
        ]\,
        p(y)
        \nonumber
        \\
        &\leq 
        \Big[
        1+{\rm R_F}(\rho)
        \Big] 
        \max_{
            \mathbb{N}_{B|Y}
            \preceq \,
            \mathbb{M}_{A|X}
        }
        \sum_{b,y} 
        \tr[
        N_{b|y}
        \Phi_{b|y}
        (\sigma^*)
        ]\,
        p(y)
        \nonumber\\
        &\leq \Big[1+{\rm R_F}(\rho)\Big]
        \max_{\sigma \in {\rm F}}
        \max_{
            \mathbb{N}_{B|Y}
            \preceq 
            \mathbb{M}_{A|X}
        }
        \sum_{b,y} 
        \tr[
        N_{b|y}
        \Phi_{b|y}(\sigma)
        ]
        p(y)
        \nonumber\\
        & = \Big[1+{\rm R_F}(\rho)\Big]
        \max_{\sigma \in {\rm F}}
        \max_{
            \mathcal{S}
        }
        \sum_{b,y} 
        \tr
        \left[
        \left(
        \sum_{a,x,\mu}
        p(b|a,y,\mu)\,
        p(x|y,\mu)\,
        p(\mu)\,
        M_{a|x}
        \right)
        \Phi_{b|y}
        (\sigma)
        \right]
        p(y)
        \nonumber\\
        &\leq 
        \Big[1+{\rm R_F}(\rho)\Big]
        \Big[1+{\rm R_\mathbb{F}}(\mathbb{M}_{A|X})
        \Big] 
        \max_{\sigma \in {\rm F}}
        \max_{
            \mathcal{S}
        }
        \sum_{b,y} 
        \tr
        \left[
        \left(
        \sum_{a,x}
        p(b|a,y,\mu)\,
        p(x|y,\mu)\,
        p(\mu)\,
        M_{a|x}
        \tilde N^*_{a|x}
        \right)
        \Phi_{b|y}
        (\sigma)
        \right]
        p(y)
        \nonumber\\
        & =
        \Big[1+{\rm R_F}(\rho)\Big]
        \Big[1+{\rm R_\mathbb{F}}(\mathbb{M}_{A|X})
        \Big]
        \max_{\sigma \in {\rm F}}
        \max_{
            \mathbb{\dbtilde{N}}_{B|Y}
            \preceq \,
            \mathbb{\tilde N}^*_{A|X}
        }
        \sum_{b,y} 
        \tr
        \left[
        \dbtilde{N}_{b|y}
        \Phi_{b|y}
        (\sigma)
        \right]
        p(y)
        \nonumber\\
        &\leq 
        \Big[1+{\rm R_F}(\rho)\Big]
        \Big[1+{\rm R_\mathbb{F}}(\mathbb{M}_{A|X})
        \Big]
        \max_{\sigma \in {\rm F}}
        \max_{
            \mathbb{\tilde N}_{A|X} \in \mathbb{F}
        }
        \max_{
            \mathbb{\dbtilde{N}}_{B|Y}
            \preceq \,
            \mathbb{\tilde N}_{A|X}
        }
        \sum_{b,y} 
        \tr
        \left[
        \dbtilde{N}_{b|y}
        \Phi_{b|y}
        (\sigma)
        \right]
        p(y)
        \nonumber\\
        &=
        \Big[
        1+{\rm R_F}(\rho)
        \Big]
        \Big[
        1+{\rm R_\mathbb{F}}
        (\mathbb{M}_{A|X})
        \Big]
        \max_{\sigma \in {\rm F}}
        \max_{
            \mathbb{\tilde N}_{A|X}
            \in 
            \mathbb{F}
        }
        P^{\rm D}_{\rm succ}(
        p_Y,
        \Phi_{B|Y},
        \sigma,
        \mathbb{\tilde{N}}_{A|X}
        ).
        \label{eq:upperbound}
    \end{align}
    In the first inequality we use $\Phi_{b|y}(\rho)\leq [1+{\rm R_F}(\rho)] \Phi_{b|y}(\sigma^*)$, $\forall x$, which follows from $\rho \leq [1+{\rm R_F}(\rho)]\sigma^*$ ($\sigma^*$ the free state from the definition of the generalised robustness) and $\Phi_{b|y}(\cdot)$ being positive $\forall b,y$. In the second inequality we maximise over all free states. In the third inequality, we use $M_{a|x} \leq [1+{\rm R_\mathbb{F}}( \mathbb{M}_{A|X} )] \tilde N_{a|x}^*$, $\forall a,x$, $N_{A|X}^*$ the free measurement assemblage from the definition of the generalised robustness. In the fourth inequality we maximise over all free measurement assemblages.
\end{proof}
	
We now prove the upper bound is achievable by using the dual conic programs. We also use the following CPosP operation. Given an arbitrary POVM $\mathbb{N}=\{N_a\}$ with $a\in\{1,...,K+N\}$, $N$ and $K$ integers, we then construct the POVM $ \mathbb{\tilde N}=\{\tilde N_x\}$ with $K$ elements as:
\begin{align}
    \nonumber \tilde N_x& 
    \coloneqq 
    N_x, \hspace{1cm}
    x\in \{1,...,K-1\},\\
    \tilde N_K& 
    \coloneqq 
    N_k
    +
    \sum_{y=K+1}^{K+N}
    N_y
    . 
    \label{eq:Ntilde}
\end{align}
This constitutes a POVM and the operation transforming $\mathbb{N}$ into $\mathbb{\tilde N}$ is a CPP operation on the initial measurement $\mathbb{N}$. In short, that any outcome of $\mathds{N}$ greater or equal than $K$ is now declared as outcome $K$.
	
\begin{proof}(Achievability)
    We start by considering a pair $(\rho,\mathbb{M}_{A|X})$. Using the dual conic formulations, there exist an operator $Z^\rho$ satisfying the conditions \eqref{eq:RoRs1}, \eqref{eq:RoRs2}, \eqref{eq:RoRs3} and a set of operators $\{Z^{\mathbb{M}_{A|X}}_{a|x}\} $, $x=1,...,\kappa$, $a=1,...,l$, satisfying \eqref{eq:RoRm1}, \eqref{eq:RoRm2}, and \eqref{eq:RoRm3}. Consider also any PMF $p_X$. We now define a set of maps $\{\Phi^{ ( \rho, \mathbb{M}_{A|X},p_X)}_{a|x} (\cdot)\}$ such that for any state $\eta$:
    \begin{align*}
        \Phi^{
            (
            \rho,
            \mathbb{M}_{A|X},
            p_X
            )
        }_{a|x}
        (\eta)
        &\coloneqq
        \alpha^{
            (
            \rho,
            \mathbb{M}_{A|X},
            p_X
            )
        }
        \tr[Z^\rho \eta]
        Z^{\mathbb{M}_{A|X}}_{a|x}
        p(x)^{-1}
        ,\\
        \alpha^{
            (
            \rho,\mathbb{M}_{A|X},p_X
            )
        }
        &\coloneqq 
        \frac{1}{
            \norm{Z^\rho}_{\mathrm{op}}
            \tr
            [
            Z^{\mathbb{M}_{A|X}}
            ]
        },
        \hspace{0.5cm}
        Z^{\mathbb{M}_{A|X}}
        \coloneqq
        \sum_{x=1}^{k}
        \sum_{a=1}^{l} Z^{\mathbb{M}_{A|X}}_{a|x}
        p(x)^{-1},
    \end{align*}
    with $\norm{\cdot}_\mathrm{op}$ the operator norm. 
    For simplicity, we use the notation $\alpha \equiv \alpha^{(\rho, \mathbb{M}_{A|X}, p_X)}$. We can check that these maps are completely-positive, linear, and that they satisfy that $\forall \eta$, $\forall x$:
    \begin{align}\label{F}
        F_x(\eta) 
        \coloneqq 
        \tr \left[
        \sum_{a=1}^{l} 
        \Phi^{
            (
            \rho,
            \mathbb{M}_{A|X},
            p_X
            )
        }_{a|x}
        (\eta)
        \right]
        =
        \frac{
            \tr[Z^\rho \eta]
        }{
            \norm{Z^\rho}_\mathrm{op}
        }
        \frac{
            \tr
            \left[
            \sum_{a=1}^{l} Z^{\mathbb{M}_{A|X}}_{a|x}
            p(x)^{-1}
            \right]
        }{
            \tr
            \left[
            Z^{\mathbb{M}_{A|X}}
            \right]
        }
        \leq 1.
    \end{align}
    The inequality follows from the variational characterisation of the operator norm $\norm{C}_{\mathrm{op}} = \max_{\rho\in\mathcal{D}(\hi)}\{|\tr[C\rho]|\}$ for any Hermitian operator $C$ \cite{Wilde_book}. We now define an instrument set as follows. Given a pair $(\rho, \mathbb{M}_{A|X})$, $\mathbb{M}_{A|X} = \{M_{a|x}\}$, $x=1,...,\kappa$, $a=1,...,l$, and an integer $J \geq 1$, we define the set of subchannels given by $\Psi^{(\rho, \mathbb{M}_{A|X}, p_X, J)} = \{\Psi^{ (\rho, \mathbb{M}_{A|X}, p_X, J)}_{b|y}(\cdot) \}$, $y=1,...,\kappa$, $b=1,...,l+J$, as:
    \begin{align}
        \Psi^{
            (\rho,
            \mathbb{M}_{A|X}, p_X
            J)
        }_{b|y}
        (\eta)
        \coloneqq
        \bigg \{
        \begin{matrix}
            \alpha
            \tr[Z^\rho \eta] 
            Z^{\mathbb{M}_{A|X}}_{b|y}
            p(y)^{-1}
            , 
            & \hspace{-1.1cm}
            b=1,...,l\\
            \frac{1}{J}
            [1-F_y(\eta)]
            \chi, 
            & 
            \hspace{0.1cm} 
            b=l+1,...,l+J
        \end{matrix}
        \label{eq:scgameD}
    \end{align} 
    with $\chi$ begin an arbitrary quantum state $\chi\geq 0$, $\tr[\chi]=1$. We can check that this is a well-defined set of subchannels because they add up to a CPTP linear map as follows. We have that $\forall J,\forall \eta, \forall y$:
    \begin{align*}
        \tr
        \left[
        \sum_{b=1}^{l+J} 
        \Psi^{ 
            (\rho,
            \mathbb{M}_{A|X}, p_X,
            J) 
        }_{b|y}
        (\eta)
        \right]
        &=
        \tr
        \left[
        \sum_{b=1}^{l} 
        \Psi^{ 
            (\rho,
            \mathbb{M}_{A|X}, p_X,
            J) 
        }_{b|y}
        (\eta)
        \right]
        +
        \tr
        \left[
        \sum_{b=l+1}^{l+J}
        \Psi^{ 
            (\rho,
            \mathbb{M}_{A|X}, p_X,
            J) 
        }_{b|y}
        (\eta)
        \right],
        \nonumber \\
        &=
        \tr
        \left[
        \sum_{b=1}^{l} 
        \alpha
        \tr[Z^\rho \eta] 
        Z^{\mathbb{M}_{A|X}}_{b|y}
        p(y)^{-1}
        \right]
        +
        \tr
        \left[
        \sum_{b=l+1}^{l+J} 
        \frac{1}{J}
        [1-F_y(\eta)]
        \chi
        \right],
        \nonumber \\
        &=
        \alpha
        \tr[Z^\rho \eta] 
        \tr
        \left[
        \sum_{b=1}^{l} 
        Z^{\mathbb{M}_{A|X}}_{b|y}
        p(y)^{-1}
        \right]
        +
        \frac{1}{J}
        [1-F_y(\eta)]
        \sum_{b=l+1}^{l+J} 
        \tr
        \left[
        \chi
        \right],
        \nonumber \\
        &=
        \alpha
        \tr[Z^\rho \eta] 
        \tr
        \left[
        \sum_{b=1}^{l} 
        Z^{\mathbb{M}_{A|X}}_{b|y}
        p(y)^{-1}
        \right]
        +
        [1-F_y(\eta)],
        \nonumber \\
        &=
        F_y(\eta)
        +
        [1-F_y(\eta)],
        \nonumber \\
        &=
        1
        .
    \end{align*}
    We now analyse the multi-object subchannel discrimination game with prior information given by $\Psi^{(\rho,\mathbb{M}_{A|X}, p_X, J)}_{B|Y}$ and a PMF $p_Y$ which we specify as $p_Y=p_X$, which can be done since $|Y| = |X| = \kappa$. We start by analysing the best \emph{fully free} player:
    \begin{align*}
        \max_{
            \substack{
                \sigma \in {\rm F}\\
                \mathbb{N}_{A|X} \in \mathbb{F}
            }
        }
        P^{\rm D}_{\rm succ}
        (
        p_Y
        ,
        \Psi^{(\rho, \mathbb{M}_{A|X}, p_X, J)}_{B|Y},
        \sigma
        ,
        \mathbb{N}_{A|X}
        )    
        =
        \max_{
            \substack{
                \sigma \in {\rm F}\\
                \mathbb{N}_{A|X} \in \mathbb{F}\\
                \mathbb{\tilde N}_{B|Y} \preceq \,
                \mathbb{N}_{A|X}
            }
        }
        \sum_{y=1}^{\kappa}
        \sum_{b=1}^{l+J}
        \tr
        \left[
        \tilde N_{b|y} 
        \Psi^{(\rho, \mathbb{M}_{A|X}, p_X, J)}_{b|y}
        (\sigma)
        \right]
        p(y).
    \end{align*} 
    We are considering QRTs of POVM sets closed under CProP, so the optimal set $\{\tilde N_{b|y}\}$ is a free object. Let us now consider, without loss of generality, that this maximisation is achieved by the \emph{fully free} pair $(\sigma^*,\mathbb{N}^*_{B|Y})$. We have:
    \begin{align}
        &\max_{
            \substack{
                \sigma \in {\rm F}\\
                \mathbb{N}_{A|X} \in \mathbb{F}
            }
        }
        P^{\rm D}_{\rm succ} (
        p_Y
        ,
        \Psi^{(\rho, \mathbb{M}_{A|X}, p_X, J)}_{B|Y},
        \sigma
        ,
        \mathbb{N}_{A|X}
        )
        \\
        &=
        \sum_{y=1}^{\kappa}
        \sum_{b=1}^{l+J}\tr
        \left[
        N^*_{b|y} \Psi^{(\rho, \mathbb{M}_{A|X}, p_X, J)}_{b|y}
        (\sigma^*)
        \right]
        p(y),
        \\
        &=
        \alpha
        \tr[Z^\rho \sigma^*] 
        \sum_{y=1}^{\kappa}
        \sum_{b=1}^{l} 
        \tr
        \left[
        N^*_{b|y} 
        Z^{\mathbb{M}_{A|X}}_{b|y}
        \right]
        +
        \sum_{y=1}^{\kappa}
        \sum_{b=l+1}^{l+J} 
        \frac{1}{J}
        \left[
        1-F_y(\sigma^*)
        \right]
        \tr
        \left[
        N^*_{b|y}
        \chi
        \right]
        p(y)
        .
        \label{eq:eq10}
    \end{align}
    In the second equality we have replaced the subchannel game \eqref{eq:scgameD}. The first term can be upper bounded as:
    \begin{align*}
        \sum_{y=1}^{\kappa}
        \sum_{b=1}^{l} 
        \tr
        \left[
        N^*_{b|y}
        Z^{\mathbb{M}_{A|X}}_{b|y}
        \right]
        \leq 
        \sum_{y=1}^{\kappa}
        \sum_{b=1}^{l} 
        \tr
        \left[
        \tilde N^*_{b|y} 
        Z^{\mathbb{M}_{A|B}}_{b|y}
        \right]
        \leq
        1,
    \end{align*}
    with the POVM $\mathbb{\tilde N}^*_y$ (with $l$ outcomes) constructed from the POVM $\mathbb{N}^*_y$ (which has $l+J$ outcomes), $\forall y=1,...,\kappa$, as defined in \eqref{eq:Ntilde}. The first inequality follows from the definition of the POVMs $\mathbb{\tilde N}^*_y$ \eqref{eq:Ntilde}.  In the second inequality we use the fact that $\mathbb{\tilde N}^*_{B|Y}$ is a free set of POVMs (because it was constructed from a free set of POVMs $\mathbb{N}^*_{B|Y}$ and a CPP operation) and therefore we can use the conic program condition \eqref{eq:RoRm3}. We now also use the fact that $1-F_y(\eta)\leq 1$, $\forall \eta$, $\forall y$, as well as \eqref{eq:RoRs3} and so equation \eqref{eq:eq10} becomes:
    \begin{align*}
        \max_{
            \substack{
                \sigma \in {\rm F}\\
                \mathbb{N}_{A|X} \in \mathbb{F}
            }
        }
        P^{\rm D}_{\rm succ}(
        p_Y
        ,
        \Psi^{(\rho, \mathbb{M}_{A|X}, p_X, J)}_{B|Y},
        \sigma
        ,
        \mathbb{N}_{A|X}
        )
        &\leq 
        \alpha 
        + 
        \frac{1}{J}
        \sum_{y=1}^{\kappa}
        p(y)
        \sum_{b=l+1}^{l+J} 
        \tr
        \left[
        N^*_{b|y}
        \chi
        \right].
    \end{align*}
    The second term in the latter expression can now be upper bounded as:
    \begin{align*}
        \sum_{b=l+1}^{l+J} 
        \tr[
        N^*_{b|y}
        \chi
        ] 
        \leq  
        \sum_{b=1}^{l+J} 
        \tr[
        N^*_{b|y}
        \chi
        ]
        &=
        \tr 
        \left[
        \left(
        \sum_{b=1}^{l+J} 
        N^*_{b|y}
        \right)
        \chi
        \right]
        =
        1
        ,
        \hspace{0.5cm}
        \forall y.
    \end{align*}
    The inequality follows because we added $l$ non-negative terms and the last equality follows from $\mathbb{N}^*_y$ being a POVM, $\sum_{b=1}^{l+J} N_{b|y}^* =\mathds{1}$, $\forall y$, and $\chi$ being a quantum state. We then get:
    \begin{align*}
        \max_{
            \substack{
                \sigma \in {\rm F}\\
                \mathbb{N}_{A|X} \in \mathbb{F}
            }
        }
        P^{\rm D}_{\rm succ}
        (
        p_Y,
        \Psi^{
            (
            \rho,
            \mathbb{M}_{A|X}, p_X,
            J
            )
        }
        ,
        \sigma
        ,
        \mathbb{N}_{A|X}
        )
        \leq 
        \alpha 
        +
        \frac{
            1
        }{
            J
        }.
    \end{align*}
    We now choose the subchannel game given by $\Psi^{ (\rho, \mathbb{M}_{A|X}, p_X, J\rightarrow \infty)}$ and therefore we get:
    \begin{align}
        \max_{
            \substack{
                \sigma \in {\rm F}\\
                \mathbb{N}_{A|X} \in \mathbb{F}
            }
        }
        P^{\rm D}_{\rm succ}
        \left(
        p_Y,
        \Psi^{
            (
            \rho
            ,
            \mathbb{M}_{A|X}
            , p_X,
            J\rightarrow \infty
            )
        },
        \sigma,
        \mathbb{N}_{A|X}
        \right)
        \leq 
        \alpha.
        \label{eq:ineqD2}
    \end{align}
    We now analyse the probability of success of a player using the fully resourceful pair $(\rho, \mathbb{M}_{A|X})$:
    \begin{align}
        P^{\rm D}_{\rm succ}
        (p_Y,
        \Psi^{(\rho
            ,
            \mathbb{M}_{A|X}, p_X,
            J
            )}_{B|Y},
        \rho
        ,
        \mathbb{M}_{A|X}
        )
        &= \nonumber
        \max_{
            \mathbb{N}_{B|Y}
            \preceq \,
            \mathbb{M}_{A|X}
        }
        \sum_{y=1}^{\kappa}
        \sum_{b=1}^{l+J} 
        \tr[
        N_{b|y}
        \Psi^{
            (\rho,
            \mathbb{M}_{A|X},
            p_X,
            J)
        }_{b|y}
        (\rho)
        ]\,
        p(y)
        \\ \nonumber
        &\geq
        \sum_{y=1}^{\kappa}
        \sum_{b=1}^{l} 
        \tr[
        M_{b|y}
        \Psi^{ (\rho, \mathbb{M}_{A|X}, 
            p_X, 
            J)}_{b|y}
        (\rho)]\,
        p(y)
        \\
        & = 
        \alpha
        \tr[
        Z^\rho 
        \rho
        ] 
        \sum_{y=1}^{\kappa}
        \sum_{b=1}^{l}
        \tr
        \left[
        M_{b|y}
        Z^{\mathbb{M}_{A|X}}_{b|y}
        \right]
        \nonumber
        \\
        &=
        \alpha 
        \Big[ 
        1+{\rm R_F}
        (\rho) 
        \Big]
        \Big[ 
        1+{\rm R_\mathbb{F}}
        (\mathbb{M}_{A|X}) 
        \Big].
        \label{eq:b4r0}
    \end{align}
    The inequality follows because one can choose to simulate the specific measurement, i.e. $N_{b|y} = M_{b|y}$ for $b \leq l$ and $N_{b|y} = 0$ for $l<b < J$. We have replaced the subchannel discrimination game with \eqref{eq:scgameD}. The last line follows from \eqref{eq:RoRs1} and \eqref{eq:RoRm1}.  We now choose the subchannel game given by $\Psi^{(\rho,\mathbb{M}_{A|X}, p_X, J\rightarrow \infty)}$ and have:
    {\small\begin{align}
        P^{\rm D}_{\rm succ}
        (p_Y,
        \Psi^{(\rho
            ,
            \mathbb{M}_{A|X}, p_X,
            J\rightarrow \infty
            )}_{B|Y},
        \rho
        ,
        \mathbb{M}_{A|X}
        )
        \geq
        \max_{
            \substack{
                \sigma \in {\rm F}\\
                \mathbb{N}_{A|X} \in \mathbb{F}
            }
        }
        P^{\rm D}_{\rm succ}
        \left(
        p_Y,
        \Psi^{
            (
            \rho
            ,
            \mathbb{M}_{A|X},
            p_X,
            J\rightarrow \infty
            )
        }_{B|Y},
        \sigma,
        \mathbb{N}_{A|X}
        \right)
        \Big[ 
        1+{\rm R_F}
        (\rho) 
        \Big]
        \Big[ 
        1+{\rm R_\mathbb{F}}
        (\mathbb{M}_{A|X}) 
        \Big].
        \label{eq:b4r2}
    \end{align}}
    Putting together \eqref{eq:b4r2} and \eqref{eq:upperbound} achieves:
    \begin{align*}
        \frac{
            P^{\rm D}_{\rm succ}(
            p_Y,
            \Psi^{
                (\rho,
                \mathbb{M}_{A|X}, p_X,
                J\rightarrow \infty
                )
            }_{B|Y}, 
            \rho
            ,
            \mathbb{M}_{A|X}
            )
        }{
            \displaystyle
            \max_{
                \substack{
                    \sigma \in {\rm F}\\
                    \mathbb{N}_{A|X} \in \mathbb{F}
                }
            }
            P^{\rm D}_{\rm succ}(
            p_Y, 
            \Psi^{
                (
                \rho,
                \mathbb{M}_{A|X}, p_X,
                J\rightarrow \infty
                )
            }_{B|Y}, 
            \sigma, \mathbb{N}_{A|X})
        }
        = \hspace{-0.1cm}
        \Big[ 
        1+{\rm R_F}(\rho) 
        \Big]
        \Big[ 
        1+{\rm R_\mathbb{F}}(\mathbb{M}_{A|X}) 
        \Big].
    \end{align*}
    This shows the upper bound is achievable thus completing the proof.
\end{proof}
		
\section{Proof of \cref{r:r2}}
\label{a:r2}

Following the same line of reasoning as in \cref{a:r1} and the preliminaries as in \cref{a:pre}, here we prove the lower bound and the achievability of the statement in \cref{r:r2} (Eq. \eqref{eq:r2}):
\begin{equation}
    \min_{
    \{p_Y,\Phi_{B|Y}\}
    } 
    \frac{P^{\mathrm{E}}_{\mathrm{err}}(p_Y,\Phi_{B|Y} ,  \rho,\mathbb{M}_{A|X})}{\min_{
    \substack{
    \sigma \in {\rm F}\\
    \mathbb{N}_{A|X} \in \mathbb{F}
    }
    }P^{\mathrm{E}}_{\mathrm{err}}( p_Y ,\Phi_{B|Y} , \sigma,\mathbb{N}_{A|X})} = \left[1- {\rm W_F}(\rho)\right] \left[1- {\rm W}_{\mathbb{F}}(\mathbb{M}_{A|X})\right] ,\nonumber
\end{equation}
where the probability of error in quantum subchannel exclusion with prior information reads:
\begin{align} 
    P^{\mathrm{E}}_{\mathrm{err}}
    (p_Y,\Phi_{B|Y},  \rho,\mathbb{M}_{A|X} ) 
    &\coloneqq   
    \min_{\mathcal{S}} 
    \sum_{g,a,x,z, b,y} 
    \delta_{g,b}\,
    s(g|a,y, z)
    \tr[M_{a|x}
    \Phi_{b|y}(\rho)] 
    r(x|y,z)\,
    q(z)\,
    p(y) 
    \\
    &=  
    \min_{\mathbb{N}_{B|Y}
    \preceq \mathbb{M}_{A|X}} 
    \sum_{b,y} 
    \tr 
    \left[ N_{b|y} 
    \Phi_{b|y}
    (\rho) 
    \right]    
    p(y) 
    ,
\end{align}
with the minimisation over all the possible strategies $\mathcal{S}$ and POVM sets $\mathbb{N}_{B|Y}$ simulable by $\mathbb{M}_{A|X}$, respectively. 

\begin{proof} (Lower bound) Given any game $(p_Y,\Phi_{B|Y})$ and any pair $(\rho, \mathbb{M}_{A|X} )$, we have:
\begin{align}
& P^{\mathrm{E}}_{\mathrm{err}}(p_Y,\Phi_{B|Y},  \rho,\mathbb{M}_{A|X}) \nonumber \\
    &= \min_{\mathbb{N}_{B|Y}\preceq \mathbb{M}_{A|X}} \sum_{b,y} \tr \left[  N_{b|y} \Phi_{b|y}(\rho)  \right]    p(y) \nonumber \\
	&\geq  [1-{\rm W_F}(\rho)] \min_{\sigma \in F} \min_{\mathbb{N}_{B|Y}\preceq \mathbb{M}_{A|X}}  \sum_{b,y} \tr\left[  N_{b|y} \Phi_{b|y}(\sigma)  \right]    p(y) , \nonumber \\
	&=  
    [1-{\rm W_F}(\rho)]  
    \min_{\sigma \in \mathrm{F}} \min_{S} \sum_{b,y} 
    \tr
    \left[   
    \sum_{a,x,z} 
    s(b|a,y, z)\, 
    M_{a|x}
    r(x|y, z)\,
    q(z)\,
    \Phi_{b|y}
    (\sigma) 
    \right]    
    p(y) 
    \nonumber \\
   &\geq  
   [1-{\rm W_F}(\rho)]
   [1-{\rm W_F}(\mathbb{M}_{A|X})]   \min_{\sigma \in \mathrm{F}} \min_{\tilde{\mathbb{N}}_{A|X}\in\mathbb{F}} \min_{S} \sum_{b,y} \tr\left[   \sum_{a,x,z} 
   s(b|a,y, z)\, 
   \tilde{N}_{a|x}\,
   r(x|y, z)\,
   q(z)\,
   \Phi_{b|y}(\sigma)  
   \right]    
   p(y)	
   \nonumber \\
   &=  [1-{\rm W_\mathrm{F}}(\rho)][1-{\rm W_F}(\mathbb{M}_{A|X})] \min_{\sigma \in F} \min_{\mathbb{\tilde{N}_{A|X} \in F}} \min_{\tilde{ \tilde{\mathbb{N}}}_{B|Y}\preceq \tilde{\mathbb{N}}_{A|X}} \sum_{b,y} \tr\left[ \tilde{\tilde{N}}_{b|y}  \Phi_{b|y}(\sigma)  \right]    p(y) \nonumber \\
   &=  [1-{\rm W_\mathrm{F}}(\rho)][1-{\rm W_F}(\mathbb{M}_{A|X})]  \min_{\sigma \in F} \min_{\mathbb{\tilde{N}_{A|X} \in F}}  P^{\mathrm{E}}_{\mathrm{err}}(p_Y,\Phi_{B|Y},  \sigma, \mathbb{\tilde{N}}_{A|X} ) 
   .
   \label{eq:lowerbound}
\end{align}
From the definition of the weight of the resource (Eq. \eqref{eq:WoRs}) for quantum states and we have: $\rho = {\rm W_F}(\rho) \rho^{G} +(1-{\rm W_F}(\rho))\sigma^*$ for ${\rm W_F}(\rho)$ and $\sigma^*\in\mathrm{F}$ satisfying the minimum in the above definition. 
The positivity of $\Phi_{b|y}$ then implies $\Phi_{b|y}( \rho) \geq (1-{\rm W_F}(\rho)) \Phi_{b|y}(\sigma^*)$, which shows the first inequality. 
The equivalent definition of weight for POVM sets leads to the inequality $M_{a|x}\geq(1-{\rm W}_{\mathbb{F}}(\mathbb{M}_{A|X}))\tilde{N}_{a|x}^*$ for ${\rm W}_{\mathbb{F}}(\mathbb{M}_{A|X})$ and $\tilde{\mathbb{N}}^*_{A|X}=\{\tilde{N}^*_{a|x}\}\in\mathbb{F}$ that satisfies the minimum weight. 
Thus the second inequality is proved.
\end{proof}

\begin{proof} (Achievability) For a pair \((\rho, \mathbb{M}_{A|X})\), the dual cone formulations assures the existence of an operator \( Y^\rho \) satisfying conditions \eqref{eq:WoRs3} and a set of operators \( \{Y^{\mathbb{M}_{A|X}}_{a|x}\} \) ($x = 1, ..., \kappa$, $a = 1, ..., l$) satisfying \eqref{eq:WoRm3}. We define a set of subchannels \( \{\Phi^{(\rho, \mathbb{M}_{A|X}, p_X)}_{a|x}\}_{a,x} \) such that for any state \( \eta \):
\begin{align}
   &\Phi^{(\rho, \mathbb{M}_{A|X}, p_X)}_{a|x}(\eta) := \beta^{(\rho, \mathbb{M}_{A|X}, p_X)} \, \tr\left[Y^\rho \eta\right] \, Y^{\mathbb{M}_{A|X}}_{a|x} \, p(x)^{-1}, \nonumber\\
   &\beta\equiv\beta^{(\rho, \mathbb{M}_{A|X}, p_X)} := \frac{1}{2 \|Y^\rho\|_{\text{op}} \, \tr\left[Y^{\mathbb{M}_{A|X}}\right]},\;\; Y^{\mathbb{M}_{A|X}} := \sum_{x=1}^{\kappa} \sum_{a=1}^{l} Y^{\mathbb{M}_{A|X}}_{a|x} \, p(x)^{-1} .\nonumber
\end{align}
We can verify that these maps are completely positive and linear, and that they satisfy:
\begin{align}
G_x(\eta) 
\coloneqq
\tr 
\left[ 
\sum_{a=1}^{l} 
\Phi^{
(
\rho,\mathbb{M}_{A|X}, p_X
)
}_{a|x}
(\eta) 
\right]  
=
\frac{
\tr 
\left[ 
Y^\rho \rho 
\right] }{ 2\|Y^\rho\|_{\text{op}}} 
\frac{
\tr
\left[ 
\sum_{a=1}^{l} 
Y^{\mathbb{M}_{A|X}}_{a|x} \, p(x)^{-1} \right] 
}{
\tr 
\left[
Y^{\mathbb{M}_{A|X}} 
\right]} 
\leq
\frac{1}{2},
\hspace{0.5cm}
\forall\, \eta,  x.
\label{eq:G}
\end{align}
Given a pair \((\rho, \mathbb{M}_{A|X})\), where \( \mathbb{M}_{A|X} = \{M_{a|x}\},~ x = 1, \ldots, \kappa,~ a = 1, \ldots, l \), we now define an instrument set
$
\Psi^{(\rho, \mathbb{M}_{A|X}, p_X)}_{B|Y}
=
\{\Psi_{b|y}^{(\rho, \mathbb{M}_{A|X}, p_X)}(\cdot)\},~ y = 1, \ldots, \kappa,~ b = 1, \ldots, l + 1,
$ as:
\begin{align}
    \Psi_{b|y}^{(\rho, \mathbb{M}_{A|X}, p_X)}(\eta) 
    \coloneqq
    \begin{cases} 
	\beta \, \tr[Y^{\rho} \eta] Y^{\mathbb{M}_{A|X}}_{b|y} \, p(y)^{-1}, & b = 1, \ldots, l \\
	 [1 - G_y(\eta)] 
     \xi^{\mathbb{M}_{A|X}}_y,
     & b = l + 1
    \end{cases}
    ,
    \hspace{1cm}
    \xi^{\mathbb{M}_{B|Y}}_y
    \coloneqq
    \frac{
    \sum_{b=1}^l 
    p(b|y) 
    Y^{\mathbb{M}_{A|X}}_{b|y}
    }{
    \sum_{b=1}^l 
    p(b|y) 
    \tr 
    [Y^{\mathbb{M}_{A|X}}_{b|y}]
    }  
    ,
    \label{eq:xi}
\end{align}
with \(\{p(b|y)\}_{b,y}\) an arbitrary conditional PMF. We can verify that this is a well-defined instrument set because they sum up to a CPTP map as follows, $\forall  \eta$, $\forall  y$:
\begin{align*}
    \tr \left[ \sum_{b=1}^{l+1} \Psi_{b|y}^{(\rho, \mathbb{M}_{A|X}, p_X)}(\eta) \right] 
    &= 
    \tr \left[ \sum_{b=1}^{l} \Psi_{b|y}^{(\rho, \mathbb{M}_{A|X}, p_X)}(\eta) \right] 
    +
    \tr 
    \left[ \Psi_{(l+1)|y}^{(\rho, \mathbb{M}_{A|X}, p_X)}(\eta) 
    \right] 
    \\
    &= 
    \tr
    \left[
    \sum_{b=1}^{l} 
    \Phi_{b|y}^{
    (\rho,
    \mathbb{M}_{A|X},
    p_X)
    }(\eta)
    \right] 
    +
    \tr \left[ 
    (1 - G_y(\eta)) 
    \xi^{\mathbb{M}_{A|X}}_y
    \right] 
    \\
    &= 
    G_y(\eta)
    +
    [1 - G_y(\eta)] 	
    \tr 
    \left[
    \xi^{\mathbb{M}_{A|X}}_y
    \right] 
    \\
    &=1.
\end{align*}
We now analyse the multi-object subchannel exclusion game with prior information given by \(\Psi^{(\rho, \mathbb{M}_{A \mid X}, p_X)}_{B \mid Y}\) and a PMF \(p_Y = p_X\), which can be done because \(|Y| = |X| = \kappa\). We start by addressing the best fully free player:
\begin{align}
    \min_{\sigma \in {\rm F}} 
    \min_{\mathbb{N}_{A|X} \in \mathbb{F}} P^{\mathrm{E}}_{\mathrm{ err}}
    (p_Y,
    \Psi^{
    (
    \rho,
    \mathbb{M}_{A|X},
    p_X
    )
    }_{B|Y},
    \sigma, 
    \mathbb{N}_{A|X}
    )
    = 
    \min_{\substack{
    \sigma \in {\rm F}\\
    \mathbb{N}_{A|X} \in \mathbb{F}\\
    \tilde{\mathbb{N}}_{B|Y} \preceq \mathbb{N}_{A|X}
    }
    }
    \sum_{y=1}^{\kappa} 
    \sum_{b=1}^{l+1} 
    \tr 
    \left[
    \tilde{N}_{b|y} 
    \Psi^{
    (\rho, \mathbb{M}_{A|X},p_X)
    }_{b|y} 
    (\sigma)
    \right] 
    p(y).
\end{align}
Let us consider that the minimum in the right hand side is achieved with the fully free pair $(\sigma = \sigma^*, \tilde{\mathbb{N}}_{B|Y} = \mathbb{N}^*_{B|Y})$. This can be done because the free set $\mathbb{F}$ is assumed to be closed under simulability of POVM sets. We then can write this as:
\begin{align}
    & 
    \min_{\sigma \in {\rm F}} 
    \min_{\mathbb{N}_{A|X} \in \mathbb{F}} P^{\mathrm{E}}_{\mathrm{ err}}
    (p_Y,
    \Psi^{
    (
    \rho,
    \mathbb{M}_{A|X},
    p_X
    )
    }_{B|Y},
    \sigma, 
    \mathbb{N}_{A|X}
    )
    =
    \sum_{y=1}^{\kappa}
    \sum_{b=1}^{l+1}
    \tr
    \left[
    N^*_{b|y} \Psi^{
    (
    \rho,
    \mathbb{M}_{A|X},
    p_X)
    }_{b|y}
    (
    \sigma^*
    )
    \right] 
    p(y)  
    \\
    &=
    \beta \,
    \tr[
    Y^{\rho} \sigma^*
    ] 
    \tr
    \left[
    \sum_{y=1}^{\kappa} 
    \sum_{b=1}^{l} 	
    N^*_{b|y} 
    Y^{\mathbb{M}_{A|X}}_{b|y} 
    \right] 
    +
    \sum_{y=1}^{\kappa} 
    [1 - G_y(\sigma^*)] 
    \tr
    \left[
    N^*_{(l+1)|y}\,
    \xi^{\mathbb{M}_{A|X}}_y
    \right]
    p(y)
    .
    \label{appCeq1}
\end{align}
We now introduce the POVM set $\tilde{\mathbb{N}}^*_{B|Y} = \{ \tilde{N}_{b|y}^* \}$ ($y=1,\ldots,\kappa,~b=1,\ldots,l$) with $\tilde{N}_{b|y}^{*} \coloneqq N_{b|y}^{*} + p(b|y) N_{(l+1)|y}^{*} $ via a CPosP of $\{ N_{b|y}^* \}$, where $\{p(b|y)\}$ is the PMF from the state in \eqref{eq:xi}. We now add and subtract the term $\beta \tr\left[Y^{\rho} \sigma^* \right] \sum_{y=1}^{\kappa} \sum_{b=1}^{l}  \tr \left[  p(b|y)N^{*}_{(l+1)|y} Y^{\mathbb{M}_{A|X}}_{b|y} \right]$, and so \eqref{appCeq1} can be rewritten as:
\begin{align}
    &\min_{\sigma \in {\rm F}} 
    \min_{\mathbb{N}_{A|X} \in \mathbb{F}} P^{\mathrm{E}}_{\mathrm{ err}}
    (p_Y,
    \Psi^{
    (
    \rho,
    \mathbb{M}_{A|X},
    p_X
    )
    }_{B|Y},
    \sigma, 
    \mathbb{N}_{A|X}
    )
    =\beta \, \tr[Y^{\rho} \sigma^*] \tr \left[ \sum_{y=1}^{\kappa} \sum_{b=1}^{l} 	\tilde{N}^*_{b|y} Y^{\mathbb{M}_{A|X}}_{b|y} \right]   \nonumber\\
    &+\sum_{y=1}^{\kappa} 
    [1 - G_y(\sigma^*)] 
    \tr
    \left[
    N^*_{(l+1)|y}\,
    \xi^{\mathbb{M}_{A|X}}_y
    \right]
    p(y) 
    - 
    \beta
    \tr
    \left[
    Y^{\rho}
    \sigma^*
    \right]
    \sum_{y=1}^{\kappa}
    \sum_{b=1}^{l} 
    p(b|y)
    \tr
    \left[
    N^{*}_{(l+1)|y}\, Y^{\mathbb{M}_{A|X}}_{b|y} 
    \right].
    \label{eq:bound_exclusion_beta}
\end{align}
The first term in \eqref{eq:bound_exclusion_beta} can be lower bounded by $\beta$ because $ \tr\left[Y^{\rho} \sigma^*\right] \geq 1$ and $ \sum_{y=1}^\kappa \sum_{b=1}^l \tr\left[Y^{\mathbb{M}_{A|X}}_{b|y} \tilde{N}^*_{b|y}\right] \geq 1$ for the positive semidefinite operators $Y^\rho$ and $\{Y^{\mathbb{M}_{A|X}}_{b|y}\}$ (as per \eqref{eq:WoRm3}).
Thus we have:
\begin{align}
    \min_{\sigma \in {\rm F}} 
    \min_{\mathbb{N}_{A|X} \in \mathbb{F}}
    P^{\mathrm{E}}_{\mathrm{ err}}
    (p_Y,
    \Psi^{
    (
    \rho,
    \mathbb{M}_{A|X},
    p_X
    )
    }_{B|Y},
    \sigma, 
    \mathbb{N}_{A|X}
    )
    &\geq
    \beta 
    +\sum_{y=1}^{\kappa}
    [1 - G_y(\sigma^*)] 
    \tr
    \left[
    N^*_{(l+1)|y}\,
    \xi^{\mathbb{M}_{A|X}}_y
    \right]
    p(y)
    \nonumber\\
    &
    -
    \beta
    \tr
    \left[
    Y^{\rho}
    \sigma^*
    \right]
    \sum_{y=1}^{\kappa}
    \sum_{b=1}^{l}
    p(b|y)
    \tr
    \left[
    N^{*}_{(l+1)|y}\, Y^{\mathbb{M}_{A|X}}_{b|y}
    \right].
    \label{eq:bound_exclusion_beta_geq}
\end{align}
Let us now prove that the last two terms add up to a non-negative value. The last two terms can be written as:
\begin{align}
  &\sum_{y=1}^{\kappa} 
  [1 - G_y(\sigma^*)]
  \tr
  \left[
  N^*_{(l+1)|y}\,
  \xi^{\mathbb{M}_{A|X}}_y
  \right]
  p(y)
  -
  \beta
  \tr
  \left[
  Y^{\rho}
  \sigma^*
  \right]
  \sum_{y=1}^{\kappa}
  \sum_{b=1}^{l}
  p(b|y)
  \tr
  \left[
  N^{*}_{(l+1)|y}
  Y^{\mathbb{M}_{A|X}}_{b|y}
  \right]
  \nonumber\\
  &=\tr\left[
  \sum_{y=1}^{\kappa}
  N^*_{(l+1)|y}
  \left(  [1 - G_y(\sigma^*)]   \xi^{\mathbb{M}_{A|X}}_y
  p(y)
  - \beta \tr\left[Y^{\rho} \sigma^* \right]  \sum_{b=1}^{l} p(b|y) Y^{\mathbb{M}_{A|X}}_{b|y} \right)
  \right].
\end{align}
Let us check the operator inside the brackets is positive semidefinite for all $y$. Let us write:
\begin{align}
  & 
  [1 - G_y(\sigma^*)]   
  \xi^{\mathbb{M}_{A|X}}_y
  p(y) 
  - 
  \beta 
  \tr
  \left[
      Y^{\rho}
      \sigma^*
  \right]
  \sum_{b=1}^{l} 
  p(b|y)
  Y^{\mathbb{M}_{A|X}}_{b|y}
  \nonumber\\
  &= 
  [1 - G_y(\sigma^*)] 
  \frac{
      \sum_{b'=1}^l 
      p(b'|y)
      Y^{
      \mathbb{M}_{A|X}}_{b'|y}
  }{
      \sum_{b'=1}^l 
      p(b'|y)
      \tr
      \left[
          Y^{\mathbb{M}_{A|X}}_{b'|y}
      \right]
  } 
  p(y)
  - 
  \beta
  \tr
  \left[
      Y^{\rho}
      \sigma^*
  \right]
  \sum_{b=1}^{l} 
  p(b|y)
  Y^{\mathbb{M}_{A|X}}_{b|y} 
  .
  \nonumber
\end{align}
Multiplying by the positive term 
$
A_y
\coloneqq
\sum_{b'=1}^l 
p(b'|y)
\tr
\left[
Y^{\mathbb{M}_{A|X}}_{b'|y}
\right]
$, we get:
{\small \begin{align}
    &    
  [1 - G_y(\sigma^*)] 
  \left(
      \left[
          \sum_{b'=1}^l 
          p(b'|y)
          Y^{
          \mathbb{M}_{A|X}}_{b'|y}
      \right]
  \right)
  p(y)
    - 
  \beta
  \tr
  \left[
      Y^{\rho}
      \sigma^*
  \right]
  \left(
      \left[
          \sum_{b=1}^{l} 
          p(b|y)
          Y^{\mathbb{M}_{A|X}}_{b|y} 
      \right]
  \right)
    A_y
    \nonumber
    \\
    \overset{1}{=}&
  \Bigg(
       [1 - G_y(\sigma^*)] 
       p(y)
        - 
      \beta
      \tr
      \left[
          Y^{\rho}
          \sigma^*
      \right]
      A_y
  \Bigg)
  \left(
      \left[
          \sum_{b=1}^{l} 
          p(b|y)
          Y^{\mathbb{M}_{A|X}}_{b|y} 
      \right]
  \right)
    \nonumber
\end{align}
We now analyse the coefficient inside the first brackets:
\begin{align}
       &[1 - G_y(\sigma^*)] 
       p(y)
        - 
      \beta
      \tr
      \left[
          Y^{\rho}
          \sigma^*
      \right]
      A_y
    \nonumber
    \\
    \overset{1}{=}&
    [1 - G_y(\sigma^*)] 
       p(y)
        - 
      \beta
      \tr
      \left[
          Y^{\rho}
          \sigma^*
      \right]
      \left(
        \sum_{b'=1}^l 
        p(b'|y)
        \tr
        \left[
        Y^{\mathbb{M}_{A|X}}_{b'|y}
        \right]
    \right)
    \\
    \overset{2}{=}&
    p(y)
    -
    \beta \, 
    \tr[
    Y^{\rho} 
    \sigma^*
    ] 
    \tr 
    \left[
    \sum_{b=1}^{l}
    Y^{\mathbb{M}_{A|X}}_{b|y} 
    \right]
    \, p(y)^{-1}
    p(y)
        - 
      \beta
      \tr
      \left[
          Y^{\rho}
          \sigma^*
      \right]
      \left(
        \sum_{b'=1}^l 
        p(b'|y)
        \tr
        \left[
        Y^{\mathbb{M}_{A|X}}_{b'|y}
        \right]
    \right)
    \\
    \overset{3}{=}&
    p(y)
    -
    \beta \, 
    \tr[
    Y^{\rho} 
    \sigma^*
    ] 
    \tr 
    \left[
    \sum_{b=1}^{l}
    Y^{\mathbb{M}_{A|X}}_{b|y} 
    \right]
        - 
      \beta
      \tr
      \left[
          Y^{\rho}
          \sigma^*
      \right]
      \left(
        \sum_{b'=1}^l 
        p(b'|y)
        \tr
        \left[
        Y^{\mathbb{M}_{A|X}}_{b'|y}
        \right]
    \right)
    \\
    \geq&
    p(y)
    -
    2
    \beta \, 
    \tr[
    Y^{\rho} 
    \sigma^*
    ] 
    \tr 
    \left[
    \sum_{b=1}^{l}
    Y^{\mathbb{M}_{A|X}}_{b|y} 
    \right]
    .
\end{align}}
In the first equality we replace $A_y$. In the second equality we replace $G_y(\sigma^*)$. In the third equality we reorganise. Finally, the inequality is due to the fact that we are subtracting a larger quantity. We now continue:
\begin{align}
    p(y) -  2\beta \, \tr[Y^{\rho} \sigma^*] \sum_{b=1}^{l}\tr \left[Y^{\mathbb{M}_{A|X}}_{b|y}\right]
    &\overset{1}{=}
    p(y) - \frac{\tr[Y^{\rho} \sigma^*]}{\|Y^\rho\|_{\text{op}}} \, \frac{\sum_{b=1}^{l}\tr \left[Y^{\mathbb{M}_{A|X}}_{b|y}\right]}{\tr\left[Y^{\mathbb{M}_{A|X}}\right]} \nonumber\\
    &\overset{2}{=}
    p(y) - \frac{\tr[Y^{\rho} \sigma^*]}{\|Y^\rho\|_{\text{op}}} \, \frac{\sum_{b=1}^{l}
    \tr 
    \left[
    Y^{\mathbb{M}_{A|X}}_{b|y}
    \right]}{
    \sum_{y'=1}^{\kappa}
    \sum_{b=1}^{l}
    \tr
    \left[
    Y^{\mathbb{M}_{A|X}}_{b|y'} 
    p(y')^{-1}
    \right]
    } \nonumber\\
    &\overset{3}{\geq} 
    p(y) - \frac{\tr[Y^{\rho} \sigma^*]}{\|Y^\rho\|_{\text{op}}} \, 
    \frac{
        p(y)
        \sum_{y'=1}^{\kappa}
        \sum_{b=1}^{l}
        \tr
        \left[
        Y^{\mathbb{M}_{A|X}}_{b|y'}
        p(y')^{-1}
        \right]
    }{
        \sum_{y'=1}^{\kappa}
        \sum_{b=1}^{l}
        \tr
        \left[
        Y^{\mathbb{M}_{A|X}}_{b|y'} 
        p(y')^{-1}
        \right]
    } \nonumber\\
    &\overset{4}{=} 
    p(y)
    \left(
    1 
    -
    \frac{
    \tr[Y^{\rho} \sigma^*]
    }{
    \|Y^\rho\|_{\text{op}}
    }
    \right)
    \nonumber
    \\
    &\overset{5}{\geq}
    0 .\nonumber
\end{align}
In the first line we replace $\beta$. In the second line we replace
$Y^{\mathbb{M}_{A|X}}$.
The inequality in the third line follows because 
$
\sum_{y'=1}^{\kappa}\sum_{b=1}^{l}\tr \left[Y^{\mathbb{M}_{A|X}}_{b|y'} p(y')^{-1}\right]\ge \sum_{b=1}^{l}\tr \left[Y^{\mathbb{M}_{A|X}}_{b|y} p(y)^{-1}\right]
$, $\forall y$, and so
$
p(y)
\sum_{y'=1}^{\kappa}\sum_{b=1}^{l}\tr \left[Y^{\mathbb{M}_{A|X}}_{b|y'} p(y')^{-1}\right]\ge \sum_{b=1}^{l}\tr \left[Y^{\mathbb{M}_{A|X}}_{b|y} 
\right]
$, $\forall y$. In the fourth line we reorganise.
The inequality in the fifth line follows because $\frac{\tr[Y^{\rho} \sigma^*]}{\|Y^\rho\|_{\text{op}}} \leq 1$. This then ultimately implies that:
\begin{align}
    \min_{\sigma \in F} 
    \min_{
    \mathbb{N}_{A|X}
    \in
    \mathbb{F}
    } 
    P^{\mathrm{E}}_{\text{ err}}
    (
    p_Y,
    \Psi^{
    (
    \rho,
    \mathbb{M}_{A|X},
    p_X)}_{B|Y},
    \sigma,
    \mathbb{N}_{A|X}
    )
    \geq
    \beta
    .
    \label{eq:greater_beta}
\end{align}
To complete the proof, let us analyse the exclusion game $(p_Y,\Psi^{(\rho, \mathbb{M}_{A|X}, p_X)}_{B|Y})$ with a fully resourceful pair $(\rho, \mathbb{M}_{A|X})$. The probability of error is given by:
\begin{align}
    P^{\mathrm{E}}_{\mathrm{ err}}(p_Y ,\Psi^{(\rho, \mathbb{M}_{A|X}, p_X)}_{B|Y} ,  \rho, \mathbb{M}_{A|X})
   &= \min_{\mathbb{N}_{B|Y}\preceq \mathbb{M}_{A|X}} \sum_{y=1}^{\kappa} \sum_{b=1}^{l+1} \tr \left[ N_{b|y} \Psi^{(\rho,\mathbb{M}_{A|X},p_X)}_{b|y} (\rho) \right] p(y)  \nonumber\\
   &\leq \sum_{y=1}^{\kappa} \sum_{b=1}^{l+1} \tr \left[ \tilde{M}_{b|y} \Psi^{(\rho,\mathbb{M}_{A|X},p_X)}_{b|y} (\rho) \right] p(y)  \nonumber\\
   &= \beta \, \tr\left[Y^\rho \rho\right] \, \sum_{y=1}^{\kappa} \sum_{b=1}^{l} \tr \left[ {M}_{b|y} Y^{\mathbb{M}_{A|X}}_{b|y} \right] p(y)^{-1} p(y)  \nonumber\\  
  &= \beta \, \tr\left[Y^\rho \rho\right] \, \sum_{y=1}^{\kappa} \sum_{b=1}^{l} \tr \left[ M_{b|y} Y^{\mathbb{M}_{A|X}}_{b|y} \right]. \nonumber
\end{align}
The inequality is due to the choice of $\mathbb{N}_{B|Y}=\{N_{b|y}\}=\{\tilde{M}_{b|y}\}$ ($y=1,\ldots,\kappa,~b=1,\ldots,l+1$) with 
\[
\tilde{M}_{b|y}=\left\{
\begin{aligned}
	&M_{b|y}\quad&&(b=1,\ldots,l)\\
	&0&&(b=l+1),
\end{aligned}
\right.
\]
which is a CProP of $\mathbb{M}_{A|X}=\{M_{a|x}\}$. Since $\tr\left[Y^\rho \rho\right]=1-{\rm W_F}(\rho)$ and $\sum_{y=1}^{\kappa} \sum_{b=1}^{l} \tr \left[ M_{b|y} Y^{\mathbb{M}_{A|X}}_{b|y} \right]=1-{\rm W_\mathbb{F}}(\mathbb{M}_{A|X})$ hold for the positive operators $Y^\rho$ and $\{Y^{\mathbb{M}_{A|X}}_{b|y}\}$, respectively, we obtain
\begin{align}
    P^{\mathrm{E}}_{\mathrm{ err}}(p_Y ,\Psi^{(\rho, \mathbb{M}_{A|X}, p_X)}_{b|y} ,  \rho, \mathbb{M}_{B|Y}) \le \beta\left[1-{\rm W_F}(\rho)\right]\left[1-{\rm W_\mathbb{F}}(\mathbb{M}_{A|X})\right].
  \label{eq:less_beta}
\end{align}
It follows from \eqref{eq:greater_beta} and \eqref{eq:less_beta} that the ratio of interest is upper bounded as
\begin{align}
    \frac{P^{\mathrm{E}}_{\mathrm{ err}}(p_Y ,\Psi^{(\rho, \mathbb{M}_{A|X}, p_X)}_{b|y} ,  \rho, \mathbb{M}_{B|Y})}
    {
    \displaystyle
    \min_{\sigma \in F} \min_{\mathbb{N}_{A|X} \in \mathbb{F}} 
    P^{\mathrm{E}}_{\mathrm{ err}}(p_Y ,\Psi^{(\rho, \mathbb{M}_{A|X}, p_X)}_{b|y} ,  \sigma, \mathbb{N}_{A|X})} 
   &\leq \frac{\beta [1-{\rm W_F}(\rho)][1-{\rm W_F}(\mathbb{M}_{A|X})]}
    {
    \displaystyle
    \min_{\sigma \in F} \min_{\mathbb{N}_{A|X} \in \mathbb{F}} P^{\mathrm{E}}_{\mathrm{ err}}( p_Y ,\Psi^{(\rho, \mathbb{M}_{A|X}, p_X)}_{b|y} , \sigma, \mathbb{N}_{A|X})} \nonumber\\
    &\leq \frac{\beta [1-{\rm W_F}(\rho)][1-{\rm W_F}(\mathbb{M}_{A|X})]}
    {\beta} \nonumber\\
    &= [1-{\rm W_F}(\rho)][1-{\rm W_F}(\mathbb{M}_{A|X})]
    .
    \label{eq:achievability}
\end{align}
Putting together \eqref{eq:lowerbound} and \eqref{eq:achievability}, we finally prove Result 2 \eqref{eq:r2}.
\end{proof}
	
\section{Proof of \cref{r:r3} (GPTs)}
\label{app:GPT}
In this appendix, we prove \eqref{eq:result3-1} in \cref{r:r3} following the proof for the quantum case in \cref{a:r1}.
We will omit the proof of \eqref{eq:result3-2}, but it is given similarly by generalising the argument in \cref{a:r2}. We first rephrase Lemma \ref{lem} in terms of GPTs:
\begin{lemma}\label{lem_conic prog_GP}
    Let $V_+$ be the positive cone generated by the state space $\Omega$ and $V_+^\circ$ be the dual cone.
    We can regard $V_+$ and $V_+^\circ$ as CCCs in $V=\mathit{lin}(\Omega)$, and they define orderings $\le_{V_+}$ and $\le_{V_+^\circ}$ in $V$ through $x\le_{V_+} y \underset{\mathrm{def}}{\iff} y-x\in V_+$ and $x\le_{V_+^\circ} y \underset{\mathrm{def}}{\iff} y-x\in V_+^\circ$ respectively.
    The generalised robustness of resource of a state $\omega\in\Omega$ and a measurement set  $\mathbb{E}_{A|X}=\{e_{a|x}\}$ ($x\in\{1,...,\kappa\},\,a\in\{1,...,l\}$) are given by:
    \begin{subequations}
        \begin{align}
            {\rm R_F^{GP}}\left(\omega\right)=
            \max_{z} \ \ 
            &\ang{z,\omega}-1, \label{eq:RoRs1_GP}\\
            {\rm s. t.} \quad
            & z\geq_{V_+^\circ} 0,\label{eq:RoRs2_GP}\\
            &\ang{z,\sigma}\leq 1\ \ (\forall \sigma \in {\rm F}),
            \label{eq:RoRs3_GP}
        \end{align}
    \end{subequations}

    and
    \begin{subequations}
        \begin{align}
            {\rm R_\mathbb{F}^{GP}}
            \left(
            \mathbb{E}_{A|X}
            \right)
            =
            \max_{\{z_{a,x}\}} 
            \ \ 
            &
            \sum_{x=1}^{\kappa}
            \sum_{a=1}^{l}
            \ang{e_{a|x},z_{a,x}}
            -1, 
            \label{eq:RoRm1_GP}\\
            {\rm s. t.} 
            \quad
            & z_{a,x}\geq_{V_+} 0\ \ 
            (\forall a, x), 
            \label{eq:RoRm2_GP}\\
            &
            \sum_{x=1}^{\kappa}
            \sum_{a=1}^{l}
            \ang{
                N_{a|x}, z_{a,x}
            }
            \leq
            1\ \ 
            (\forall \mathbb{N}=\{N_{a|x}\} \in \mathbb{F}).
            \label{eq:RoRm3_GP}
        \end{align}
    \end{subequations}
    These are the dual conic formulations of the generalised robustnesses for states and measurement sets.
\end{lemma} 
\begin{proof}
The proof proceeds in a similar way as Lemma \ref{lem}.
In fact, instead of the vector space $\LL_S(\hi)$, the Hilbert-Schmidt inner product $\ang{\cdot,\cdot}_{\mathrm{HS}}$, and the cone $\LL_S^+(\hi)$, here we use $V$, $\ang{\cdot,\cdot}$, and $V_+^\circ$ respectively (remember that $(V_+^\circ)^\circ=V_+$ holds).
It is easy to see that the same argument in the proof of Lemma \ref{lem} can be developed also in this case.
The strong duality is verified by the fact that the positive cone $V_+$ is generating ($\mathit{lin}(V_+)=V$, which is satisfied in the present setting) if and only if there is an interior point $z_0\in \mathrm{int}(V_+^{})$.
Multiplying $z_0$ by small $\lambda>0$, we can construct a strictly feasible solution $\{z_{a,x}\}$ with $z_{a,x}=\lambda z_0\ (\forall a,x)$ for \eqref{eq:RoRs2_GP} and \eqref{eq:RoRs3_GP}, and thus the strong duality holds.
\end{proof}
	
\begin{proof}[Proof of Result 3]
    (Upper bound)
    For any GPScD-PI $(p_Y,\Xi_{B|Y},\omega,\mathbb{E}_{A|X})$, we have
    \begin{align}
        &P^{\rm GPD}_{\rm succ}
        (p_Y,\Xi_{B|Y},\omega,\mathbb{E}_{A|X})
        \nonumber
        \\
        &
        =
        \max_{
            \mathbb{N}_{B|Y}
            \preceq \,
            \mathbb{E}_{A|X}
        }
        \sum_{b,y}
        \ang{
            N_{b|y},\,
            \xi_{b|y}
            (\omega)
        }\,
        p(y)
        \nonumber
        \\
        &\leq 
        \Big[
        1+{\rm R_F^{GP}}(\omega)
        \Big] 
        \max_{
            \mathbb{N}_{B|Y}
            \preceq \,
            \mathbb{E}_{A|X}
        }
        \sum_{b,y} 
        \ang{
            N_{b|y},\, 
            \xi_{b|y}
            (\sigma^*)
        }\,
        p(y)
        \nonumber\\
        &\leq \Big[1+{\rm R_F^{GP}}(\omega)\Big]
        \max_{\sigma \in {\rm F}}
        \max_{
            \mathbb{N}_{B|Y}
            \preceq 
            \mathbb{E}_{A|X}
        }
        \sum_{b,y} 
        \ang{
            N_{b|y},\,
            \xi_{b|y}(\sigma)
        }
        p(y)
        \nonumber\\
        & = \Big[1+{\rm R_F^{GP}}(\omega)\Big]
        \max_{\sigma \in {\rm F}}
        \max_{
            \mathcal{S}
        }
        \sum_{b,y} 
        \ang
        {
            \left(
            \sum_{a,x,\mu}
            p(b|a,y,\mu)\,
            p(x|y,\mu)\,
            p(\mu)\,
            e_{a|x}
            \right),\,
            \xi_{b|y}
            (\sigma)
        }
        p(y)
        \nonumber\\
        &\leq 
        \Big[1+{\rm R_F^{GP}}(\omega)\Big]
        \Big[1+{\rm R_\mathbb{F}^{GP}}(\mathbb{E}_{A|X})
        \Big] 
        \max_{\sigma \in {\rm F}}
        \max_{
            \mathcal{S}
        }
        \sum_{b,y} 
        \ang
        {
            \left(
            \sum_{a,x}
            p(b|a,y,\mu)\,
            p(x|y,\mu)\,
            p(\mu)\,
            \tilde N^*_{a|x}
            \right),\,
            \xi_{b|y}
            (\sigma)
        }
        p(y)
        \nonumber\\
        & =
        \Big[1+{\rm R_F^{GP}}(\omega)\Big]
        \Big[1+{\rm R_\mathbb{F}^{GP}}(\mathbb{E}_{A|X})
        \Big]
        \max_{\sigma \in {\rm F}}
        \max_{
            \mathbb{\dbtilde{N}}_{B|Y}
            \preceq \,
            \mathbb{\tilde N}^*_{A|X}
        }
        \sum_{b,y} 
        \ang
        {
            \dbtilde{N}_{b|y},\,
            \xi_{b|y}
            (\sigma)
        }
        p(y)
        \nonumber\\
        &\leq 
        \Big[1+{\rm R_F^{GP}}(\omega)\Big]
        \Big[1+{\rm R_\mathbb{F}^{GP}}(\mathbb{E}_{A|X})
        \Big]
        \max_{\sigma \in {\rm F}}
        \max_{
            \mathbb{\tilde N}_{A|X} \in \mathbb{F}
        }
        \max_{
            \mathbb{\dbtilde{N}}_{B|Y}
            \preceq \,
            \mathbb{\tilde N}_{A|X}
        }
        \sum_{b,y} 
        \ang
        {
            \dbtilde{N}_{b|y},\,
            \xi_{b|y}
            (\sigma)
        }
        p(y)
        \nonumber\\
        &=
        \Big[
        1+{\rm R_F^{GP}}(\omega)
        \Big]
        \Big[
        1+{\rm R_\mathbb{F}^{GP}}
        (\mathbb{E}_{A|X})
        \Big]
        \max_{\sigma \in {\rm F}}
        \max_{
            \mathbb{N}_{A|X}
            \in 
            \mathbb{F}
        }
        P^{\rm GPD}_{\rm succ}(
        p_Y,
        \Xi_{B|Y},
        \sigma,
        \mathbb{N}_{A|X}
        ).
        \label{eq:upperbound_GP}
    \end{align}
    In the first inequality we use $\xi_{b|y}(\omega)\leq_{V_+} [1+{\rm R_F^{GP}}(\omega)] \xi_{b|y}(\sigma^*)$, $\forall x$, which follows from $\omega \leq_{V_+} [1+{\rm R_F^{GP}}(\omega)]\sigma^*$ ($\sigma^*$ the free state from the definition of the generalised robustness) and $\xi_{b|y}(\cdot)$ being positive $\forall b,y$. 
    In the second inequality we maximise over all free states. 
    In the third inequality, we use $e_{a|x} \leq_{V_+^\circ} [1+{\rm R_\mathbb{F}^{GP}}( \mathbb{E}_{A|X} )] \tilde N_{a|x}^*$, $\forall a,x$, $\mathbb{N}_{A|X}^*$ the free measurement set from the definition of the generalised robustness. 
    In the fourth inequality we maximise over all free measurement sets.
\end{proof}
	
\begin{proof}(Achievability)
    Let $(\omega,\mathbb{E}_{A|X})$ be a pair of a state and a measurement set of a GPT $(\Omega,\eff)$.
    As we have seen in Lemma \ref{lem_conic prog_GP}, there exist an element $z^\omega\in V_+^\circ$ satisfying the conditions \eqref{eq:RoRs1_GP}, \eqref{eq:RoRs2_GP}, \eqref{eq:RoRs3_GP} and a set of elements $\{z^{\mathbb{E}_{A|X}}_{a|x}\}\subset V_+$ ($x=1,...,\kappa, a=1,...,l$) satisfying \eqref{eq:RoRm1_GP}, \eqref{eq:RoRm2_GP}, and \eqref{eq:RoRm3_GP}. 
    Let $p_X$ be a PMF with $p(x)>0$ ($\forall x$). 
    We define a set of maps $\{\xi^{ ( \omega, \mathbb{E}_{A|X},p_X)}_{a|x}(\cdot) \}$ such that for any state $\eta\in \Omega$
    \begin{equation}\label{eq:explicit}
        \begin{aligned}
            &\xi^{
                (
                \omega,
                \mathbb{E}_{A|X},
                p_X
                )
            }_{a|x}
            (\eta)
            \coloneqq
            \alpha^{
                (
                \omega,
                \mathbb{E}_{A|X}
                )
            }
            \ang{z^\omega, \eta}
            z^{\mathbb{E}_{A|X}}_{a|x}
            p(x)^{-1}
            ,\\
            &\alpha\equiv\alpha^{
                (
                \omega,\mathbb{E}_{A|X}
                )
            }
            \coloneqq 
            \frac{1}{
                \norm{z^\omega}_{u}
                \ang
                {
                    u,z^{\mathbb{E}_{A|X}}
                }
            },
            \hspace{0.5cm}
            z^{\mathbb{E}_{A|X}}
            \coloneqq
            \sum_{x=1}^{\kappa}
            \sum_{a=1}^{l} z^{\mathbb{E}_{A|X}}_{a|x}
            p(x)^{-1}.
        \end{aligned}
    \end{equation}
    In the equations, $u\in \eff$ is the unit effect, and $\norm{\cdot}_u$ is the \emph{order unit norm} \cite{Alfsen1971} in $V$ defined as
    \begin{equation}\label{eq:ou norm}
        \begin{aligned}
            \norm{z}_u:&=\inf\{\lambda\ge0\mid -\lambda u\le_{V_+^\circ} z \le_{V_+^\circ} \lambda u\}\\
            &=\sup\{|\ang{z,\omega}|\mid \omega\in\Omega\}.
        \end{aligned}
    \end{equation}
    The order unit norm is clearly a natural generalisation of the operator norm for quantum formulation, where $u=\mathds{1}$ and $\Omega=\mathcal{D}(\hi)$. 
    We can check that these maps are linear and positive, i.e., $\xi^{ (\omega, \mathbb{E}_{A|X},p_X)}_{a|x}\colon V_+\to V_+$ ($\forall a,x$).
    Also, they satisfy $\forall \eta$, $\forall x$:
    \begin{align*}
        F_x(\eta) 
        \coloneqq 
        \ang{u,\,
            \sum_{a=1}^{l} 
            \xi^{
                (
                \omega,
                \mathbb{E}_{A|X}
                )
            }_{a|x}
            (\eta)
        }
        =
        \frac{
            \ang{z^\omega, \eta}
        }{
            \norm{z^\omega}_u
        }
        \frac{
            \ang{u,\,
                \left(
                \sum_{a=1}^{l} z^{\mathbb{E}_{A|X}}_{a|x}
                p(x)^{-1}
                \right)
        }}{
            \ang{u,
                z^{\mathbb{E}_{A|X}}
        }}
        \leq 1.
    \end{align*}
    because $\ang{z^\omega, \eta}\le\norm{z^\omega}_u$ and $\langle u,\,(\sum_{a=1}^{l} z^{\mathbb{E}_{A|X}}_{a|x}p(x)^{-1})\rangle\le \ang{u,z^{\mathbb{E}_{A|X}}}$ from the definitions \eqref{eq:ou norm} and \eqref{eq:explicit} respectively.    We now construct an instrument set that realizes the maximum in \eqref{eq:result3-1} as follows. 
    Given a pair $(\omega, \mathbb{E}_{A|X})$ ($x=1,...,\kappa, a=1,...,l$) and an integer $J \geq 1$, we define $\Xi^{(\omega, \mathbb{E}_{A|X}, p_X, J)}_{B|Y} = \{\xi^{ (\omega, \mathbb{E}_{A|X}, p_X, J)}_{b|y}\}$ ($y=1,...,\kappa, b=1,...,l+J$) by
    \begin{align}
        \xi^{
            (\omega,\mathbb{E}_{A|X}, p_X,J)
        }_{b|y}
        (\eta)
        \coloneqq
        \left\{
        \begin{aligned}
            &\alpha
            \ang{z^\omega, \eta} 
            z^{\mathbb{E}_{A|X}}_{b|y}
            p(y)^{-1}\quad
            &&(b=1,...,l)\\
            &\frac{1}{J}
            [1-F_y(\eta)]
            \chi\quad
            &&(b=l+1,...,l+J)
        \end{aligned}
        \right.
        \label{eq:scgameD_GP}
    \end{align} 
    with an arbitrary state $\chi\in\Omega$. 
    This is a well-defined instrument set: they are positive and add up to a channel because
    \begin{align*}
        \ang
        {u,\,
            \sum_{b=1}^{l+J} 
            \xi^{ 
                (\omega,
                \mathbb{E}_{A|X}, p_X,
                J) 
            }_{b|y}
            (\eta)
        }
        &=
        \ang
        {u,\,
            \sum_{b=1}^{l} 
            \xi^{ 
                (\omega,
                \mathbb{E}_{A|X}, p_X,
                J) 
            }_{b|y}
            (\eta)
        }
        +
        \ang
        {u,\,
            \sum_{b=l+1}^{l+J}
            \xi^{ 
                (\omega,
                \mathbb{E}_{A|X}, p_X,
                J) 
            }_{b|y}
            (\eta)
        }
        \nonumber \\
        &=
        \ang
        {u,\,
            \sum_{b=1}^{l} 
            \alpha
            \ang{z^\omega, \eta} 
            z^{\mathbb{E}_{A|X}}_{b|y}
            p(y)^{-1}
        }
        +
        \ang
        {u,\,
            \sum_{b=l+1}^{l+J} 
            \frac{1}{J}
            [1-F_y(\eta)]
            \chi
        }\\
        &=
        F_y(\eta)
        +
        [1-F_y(\eta)]
        \nonumber \\
        &=
        1\quad(\forall y).
    \end{align*}
    We note that for each $y$ the instrument $\Xi_y=\{\xi^{(\omega,\mathbb{E}_{A|X}, p_X,J)}_{b|y}\}_{b=1}^{l+J}$ expresses a \emph{measure-and-prepare channel} and such channel is ``completely positive" also in the framework of GPTs \cite{Plavala2023}.
    Let us analyse the GPScD-PI given by $\Xi^{(\omega,\mathbb{E}_{A|X}, p_X, J)}_{B|Y}$ and a PMF $p_Y$ which we specify as $p_Y=p_X$ ($B=\{1,\ldots,l+J\},\,Y=\{1,\ldots,\kappa\}$).
    For fully free cases, we have 
    \begin{align*}
        \max_{
            \substack{
                \sigma \in {\rm F}\\
                \mathbb{N}_{A|X} \in \mathbb{F}
            }
        }
        P^{\rm GPD}_{\rm succ}
        (
         p_Y
        ,
        \Xi^{(\omega, \mathbb{E}_{A|X}, p_X, J)}_{B|Y},
        \sigma
        ,
        \mathbb{N}_{A|X}
        )    
        =
        \max_{
            \substack{
                \sigma \in {\rm F}\\
                \mathbb{N}_{A|X} \in \mathbb{F}\\
                \mathbb{\tilde N}_{B|Y} \preceq \,
                \mathbb{N}_{A|X}
            }
        }
        \sum_{y=1}^{\kappa}
        \sum_{b=1}^{l+J}
        \ang
        {
            \tilde N_{b|y},\,
            \xi^{(\omega, \mathbb{E}_{A|X}, p_X, J)}_{b|y}
            (\sigma)
        }
        p(y).
    \end{align*} 
    Because of the compactness, we can assume that this maximisation is achieved by the fully free pair $(\sigma^*,\mathbb{N}^*_{B|Y})$. 
    We then have:
    \begin{align}
        P^{\rm GPD}_{\rm succ} (
        p_Y
        ,
        \Xi^{(\omega, \mathbb{E}_{A|X}, p_X, J)}_{B|Y},
        \sigma^*
        ,
        \mathbb{N}^*
        )
        &=
        \sum_{y=1}^{\kappa}
        \sum_{b=1}^{l+J}
        \ang
        {
            N^*_{b|y},\, \xi^{(\omega, \mathbb{N}, p_X, J)}_{b|y}
            (\sigma^*)
        }
        p(y)
        \\
        &=
        \alpha
        \ang{z^\omega, \sigma^*} 
        \sum_{y=1}^{\kappa}
        \sum_{b=1}^{l} 
        \ang{
            N^*_{b|y},\, 
            z^{\mathbb{E}_{A|B}}_{b|y}
        }
        +
        \frac{1}{J}
        \sum_{y=1}^{\kappa}
          \sum_{b=l+1}^{l+J} 
        \left[
        1-F_y(\sigma^*)
        \right]
        \ang
        {
            N^*_{b|y},\,
            \chi
        }p(y).\label{eq:eq10_GP}
    \end{align}
    To evaluate the first term, we introduce a set of $l$-outcome measurements $\mathbb{\tilde N}=\{\tilde{N}_y^*\}=\{\tilde{N}^*_{b'|y}\}$ constructed from $\mathbb{N}_{B|Y}^*$ as:
    \begin{align}
        \nonumber 
        \tilde N_{b'|y}^*& 
        \coloneqq 
        N^*_{b'|y}, \hspace{1cm}
        b'\in \{1,...,l-1\},\\
        \tilde N^*_{l|y}& 
        \coloneqq 
        N^*_{l|y}
        +
        \sum_{b=l+1}^{l+J}
        N^*_{b|y}
        . 
        \label{eq:Ntilde_GP}
    \end{align}
    We have
    \begin{align*}
        \sum_{y=1}^{\kappa}
        \sum_{b=1}^{l} 
        \ang
        {
            N^*_{b|y},\,
            z^{\mathbb{E}_{A|X}}_{b|y}
        }
        \leq 
        \sum_{y=1}^{\kappa}
        \sum_{b=1}^{l} 
        \ang
        {
            \tilde N^*_{b|y},\, 
            z^{\mathbb{E}_{A|X}}_{b|y}
        }
        \leq
        1.
    \end{align*}
    The first inequality is straightforward, and the second inequality follows from the fact that the construction \eqref{eq:Ntilde_GP} of $\mathbb{\tilde N}$ is a CPosP of $\mathbb{N}^*_{B|Y}$ and thus $\mathbb{\tilde N}\preceq \mathbb{N}^*_{B|Y}$, which enables us to use \eqref{eq:RoRm3_GP}. For the second term of \eqref{eq:eq10_GP}, we use $1-F_y(\eta)\leq 1$ ($\forall \eta$, $\forall y$).
    The relation \eqref{eq:eq10_GP} becomes:
    \begin{align*}
        P^{\rm GPD}_{\rm succ}(
        p_Y
        ,
        \Xi^{(\omega, \mathbb{E}_{A|X}, p_X, J)}_{B|Y},
        \sigma
        ,
        \mathbb{N}_{A|X}
        )
        &\leq 
        \alpha 
        + 
        \frac{1}{J}
        \sum_{y=1}^{\kappa}
        \sum_{b=l+1}^{l+J} 
        p(y)\ang
        {
            N^*_{b|y},\,
            \chi
        }.
    \end{align*}
    The r.h.s. can be upper bounded as:
    \begin{align*}
    	 \sum_{y=1}^{\kappa}
        \sum_{b=l+1}^{l+J} 
        p(y)
        \ang{
            N^*_{b|y},\,
            \chi
        }
        \leq  
         \sum_{y=1}^{\kappa}
        \sum_{b=1}^{l+J} 
        p(y)
        \ang{
            N^*_{b|y},\,
            \chi
        }
        &=
         \sum_{y=1}^{\kappa}
         p(y)
        =
        1\quad(\forall y),
    \end{align*}
    and thus
    \begin{align*}
        \max_{
            \substack{
                \sigma \in {\rm F}\\
                \mathbb{N}_{A|X} \in \mathbb{F}
            }
        }
        P^{\rm GPD}_{\rm succ}
        (
        p_Y
        ,
        \Xi^{(\omega, \mathbb{E}_{A|X}, p_X, J)}_{B|Y},
        \sigma
        ,
        \mathbb{N}_{A|X}
        )    
        \leq 
        \alpha 
        +
        \frac{1
        }{
            J
        }.
    \end{align*}
    With $J\to \infty$, we obtain
    \begin{align}
        \max_{
            \substack{
                \sigma \in {\rm F}\\
                \mathbb{N}_{A|X} \in \mathbb{F}
            }
        }P^{\rm GPD}_{\rm succ}
        \left(
        p_Y,
        \Xi^{
            (
            \omega
            ,
            \mathbb{E}_{A|X}
            , p_X,
            J\rightarrow \infty
            )
        },
        \sigma,
        \mathbb{N}_{A|X}
        \right)
        \leq 
        \alpha.
        \label{eq:ineqD2_GP}
    \end{align}
    We next investigate the fully resourceful $(\omega, \mathbb{E}_{A|X})$. 
    It holds from \eqref{eq:RoRs1_GP} and \eqref{eq:RoRm1_GP} that:
    \begin{align}
        P^{\rm GPD}_{\rm succ}
        (p_Y,
        \Xi^{(\omega
            ,
            \mathbb{E}_{A|X}, p_X,
            J
            )}_{B|Y},
        \omega
        ,
        \mathbb{E}_{A|X}
        )
        &= \nonumber
        \max_{
            \mathbb{N}_{B|Y}
            \preceq \,
            \mathbb{M}_{A|X}
        }
        \sum_{y=1}^{\kappa}
        \sum_{b=1}^{l+J} 
        \ang{
            N_{b|y},\,
            \xi^{
                (\omega,
                \mathbb{E}_{A|X},
                p_X,
                J)
            }_{b|y}
            (\omega)
        }\,
        p(y)
        \\ \nonumber
        &\geq
        \sum_{y=1}^{\kappa}
        \sum_{b=1}^{l} 
        \ang{
            e_{b|y},\,
            \xi^{ (\omega, \mathbb{E}_{A|X}, 
                p_X, 
                J)}_{b|y}
            (\omega)}\,
        p(y)
        \\
        & = 
        \alpha
        \ang{
            z^\omega, 
            \omega
        }
        \sum_{y=1}^{\kappa}
        \sum_{b=1}^{l}
        \ang{
            e_{b|y},\,
            z^{\mathbb{E}_{A|X}}_{b|y}
        }
        \nonumber
        \\
        &=
        \alpha 
        \Big[ 
        1+{\rm R_F^{GP}}
        (\omega) 
        \Big]
        \Big[ 
        1+{\rm R_\mathbb{F}^{GP}}
        (\mathbb{E}_{A|X}) 
        \Big].
        \label{eq:b4r0_GP}
    \end{align}
    With $J\to \infty$, the relations \eqref{eq:ineqD2_GP} and  \eqref{eq:b4r0_GP} imply
    \begin{equation}\label{eq:b4r2_GP}
        \begin{aligned}
            P^{\rm GPD}_{\rm succ}
            &(
            p_Y,
            \Xi^{(\omega
                ,
                \mathbb{E}_{A|X}, p_X,
                J\rightarrow \infty
                )}_{B|Y},
            \omega
            ,
            \mathbb{E}_{A|X}
            )\\
            &\qquad\qquad\geq
            \max_{
                \substack{
                    \sigma \in {\rm F}\\
                    \mathbb{N}_{A|X} \in \mathbb{F}
                }
            }P^{\rm GPD}_{\rm succ}
            (p_Y,
            \Xi^{
                (
                \omega
                ,
                \mathbb{E}_{A|X},
                p_X,
                J\rightarrow \infty
                )
            }_{B|Y},
            \sigma,
            \mathbb{N}_{A|X}
            )
            \Big[ 
            1+{\rm R_F^{GP}}
            (\omega) 
            \Big]
            \Big[ 
            1+{\rm R_\mathbb{F}^{GP}}
            (\mathbb{E}_{A|X}) 
            \Big].
        \end{aligned}
    \end{equation}
    We can conclude from \eqref{eq:upperbound_GP} and \eqref{eq:b4r2_GP}
    \begin{align*}
        \frac{
            P^{\rm GPD}_{\rm succ}(
            p_Y,
            \Xi^{
                (\omega,
                \mathbb{E}_{A|X}, p_X,
                J\rightarrow \infty
                )
            }_{B|Y}, 
            \omega
            ,
            \mathbb{E}_{A|X}
            )
        }{
            \displaystyle
            \max_{
                \substack{
                    \sigma \in {\rm F}\\
                    \mathbb{N}_{A|X} \in \mathbb{F}
                }
            }
            P^{\rm GPD}_{\rm succ}(p_Y, \Xi^{
                (
                \omega,
                \mathbb{E}_{A|X}, p_X,
                J\rightarrow \infty
                )
            }_{B|Y}, 
            \sigma, \mathbb{N}_{A|X})
        }
        = \hspace{-0.1cm}
        \Big[ 
        1+{\rm R_F^{GP}}(\omega) 
        \Big]
        \Big[ 
        1+{\rm R_\mathbb{F}^{GP}}(\mathbb{E}_{A|X}) 
        \Big],
    \end{align*}
    which completes the proof.
\end{proof}
	
\end{document}